\documentclass{fundam}

\usepackage{times}
\usepackage{enumitem}
\usepackage{graphics,graphicx}
\usepackage{epsf}
\usepackage{stmaryrd}
\usepackage{longtable}
\usepackage{latexsym}
\usepackage{amsfonts,amssymb}
\usepackage[usenames]{color}
\usepackage[mathscr]{eucal}
\usepackage{xspace}
\usepackage{algorithmic}
\usepackage{algorithm}
\usepackage{subfigure}
\usepackage{wrapfig}
\usepackage{url}
\usepackage{mathtools}
\usepackage{tikz}
\usepackage{pgf}
\usepackage{array}
\usepackage{multirow}
\usepackage{multicol}

\usepackage{main}

\usetikzlibrary{shapes,arrows,automata}
\usetikzlibrary{calc}

\newcolumntype{H}{>{\setbox0=\hbox\bgroup}c<{\egroup}@{}}

\newcommand{\ATLSi}{{\rm mv-ATL}$^*_{\ra}$\xspace}
\newcommand{\ATLi}{{\rm mv-ATL}$_{\ra}$\xspace}
\def\ttt{\mbox{\bf t}}
\def\uu{\mbox{\bf u}}
\def\ff{\mbox{\bf f}}
\def\Ltrans{LT}

\newcommand{\designated}{{\cal D}}
\renewcommand{\lessutil}{\ra}
\newcommand{\join}{\sqcup}
\newcommand{\meet}{\sqcap}
\newcommand{\bigjoin}{\bigsqcup}
\newcommand{\bigmeet}{\bigsqcap}
\newcommand{\pol}[1]{\prop{pol_{#1}}}
\newcommand{\dronepol}[1]{\prop{d\text{-}pol_{#1}}}
\newcommand{\groundpol}[1]{\prop{g\text{-}pol_{#1}}}
\newcommand{\droneok}[1]{\prop{d\text{-}ok_{#1}}}
\newcommand{\groundok}[1]{\prop{g\text{-}ok_{#1}}}
\newcommand{\dlattice}{{\bf 2 + 2 \times 2 + 2 \times 2}}

\newenvironment{myalgorithm}[1]{
  \begin{tabular}{|p{\textwidth-0.5cm}|}
  \hline
  \vspace{0.01cm}
  \textbf{Algorithm}\quad #1;
  \smallskip
  \\ \hline
  \begin{enumerate}
}{
  \end{enumerate}
  \\ \hline
  \end{tabular}
}

\begin{document}
\setcounter{page}{1}   \issue{175~(2020)}        \papernumber{1940}

\title{Multi-Valued Verification of Strategic Ability}

\author{Wojciech Jamroga, Beata Konikowska,
             Damian Kurpiewski, Wojciech Penczek\thanks{Address for correcpondence:
                       Institute of Computer Science, Polish Academy of Sciences,  Jana Kazimierza 5, 01-248 Warsaw, Poland}
 \\
Institute of Computer Science \\
Polish Academy of Sciences \\
 Jana Kazimierza 5, 01-248 Warsaw, Poland \\
 \{jamroga, konikowska, kurpiewski, penczek\}@ipipan.waw.pl
}

\maketitle

\runninghead{W. Jamroga et al.}{Multi-Valued Verification of Strategic Ability}

\vspace*{-4mm}
\begin{abstract}
Some multi-agent scenarios call for the possibility of evaluating specifications in a richer domain of truth values.
Examples include runtime monitoring of a temporal property over a growing prefix of an infinite path, inconsistency
analysis in distributed databases, and verification methods that use incomplete anytime algorithms, such as bounded model checking.
In this paper, we present \emph{multi-valued alternating-time temporal logic} (\ATLSi), an expressive logic to specify strategic
abilities in multi-agent systems.
It is well known that, for branching-time logics, a general method for model-independent translation from multi-valued
to two-valued model checking exists.
We show that the method cannot be directly extended to \ATLSi.
We also propose two ways of overcoming the problem.
Firstly, we identify constraints on formulas for which the model-independent translation can be suitably adapted.
Secondly, we pre\-sent a model-dependent reduction that can be applied to all formulas of \ATLSi.
We show that, in all cases, the complexity of verification increases only linearly when
new truth values are added to the evaluation domain.
We also consider several examples that show possible applications of \ATLSi
and motivate its use for model checking multi-agent systems.
\end{abstract}

\section{Introduction}\label{sec:intro}

Alternating-time temporal logic \ATLs and its less expressive variant \ATL~\cite{Alur97ATL,Alur02ATL} are probably the most popular
logics that allow for reasoning about agents' abilities in strategic encounters. \ATLs combines features of temporal logic and basic
game theory, encapsulated in the main language construct of the logic, $\coop{A}\gamma$, which can be read as ``the group of agents
$A$ has a strategy to enforce $\gamma$''. Property $\gamma$ can include operators $\Next$ (``next''), $\Always$ (``always''),
$\Sometm$ (``eventually'') and/or $\Until$ (``until'').
Much research on \ATLs has focused on the way it can be used for verification of multi-agent systems,
including theoretical studies on the complexity of model checking, as well as practical verification algorithms.

Multi-valued semantics have been already proposed for various temporal and temporal-epistemic logics, and applied in verification of
distributed and multi-agent systems.
In this paper, we extend the general approach of~\cite{kp02a,kp06} to verification of strategic abilities for agents and their coalitions.
Many relevant properties of multi-agent systems are intuitively underpinned by the ability (or inability) of particular agents to obtain particular outcomes.
For example, most functionality requirements can be specified as the ability of the authorized users to achieve their legitimate goals.
At the same time, many security properties can be phrased in terms of the inability of unauthorized users to compromise the system.
We show that reasoning about, and verification of such statements can be naturally conducted in richer domains of truth values, regardless of the actual notion of strategic play and constraints on the agents' epistemic capabilities.
Moreover, multiple truth values incur no significant increase of either the theoretical or the practical complexity of the model checking problem.
In fact, we show that multi-valued verification of strategic abilities can be usually done by means of an efficient translation to classical verification.
We also present a case study that demonstrates the modeling value of the approach and its efficiency.

\subsection{Multi-valued model checking: Why bother?}\label{sec:motivation}

Typically, model checking is a yes/no problem.
That is, the output is a classical truth value: $\top$ for ``yes'', $\bot$ for ``no.''
However, it is sometimes convenient to consider the output of verification in a richer domain of values.
This can be due to at least three reasons.
First, our reasoning about the world can be based on a non-classical notion of truth (e.g., graded, fuzzy, constructive, defeasible etc.).
For example, if an atomic statement \prop{pol} (``the area is polluted'') is evaluated according to measurements conducted in the environment,
the possible outputs can include ``highly,'' ``much,'' ``little,'' and ``not at all,'' instead of simply ``yes'' and ``no.''
Multi-valued semantics take those outputs directly as the possible truth values of \prop{pol}, and hence also of more complex formulas involving \prop{pol}.
This is particularly useful in case of model checking realistic systems, as it allows to lift logical reasoning to a richer domain of answers in a compact way.

Secondly, we can understand the properties of the world in a classical way (being always completely true or fully false), but our model of the system can be only partially conclusive. A good example is runtime monitoring of temporal properties, where a temporal formula (interpreted with an infinite time horizon in mind) is checked on a finite but constantly growing sequence of events, observed so far.
Consider for instance specification $\Sometm\prop{cool}$ (``the reactor will be cooled down sometime in the future'').
If \prop{cool} has already occurred then the formula is clearly true whatever happens next. What if it has not occurred? Then, the formula may still turn out true (because \prop{cool} may occur in subsequent steps), but it can also turn out false; effectively, the truth value is unknown in our current model. Likewise, formula $\Always\prop{cool}$ (``\prop{cool} will always hold'') can be only proved false in the course of monitoring, or the model is inconclusive regarding its value.
Indeed, well known approaches to runtime monitoring use multi-valued interpretation of temporal formulas in finite runs~\cite{Bauer06runtime3val,Bauer07runtime4val}.

Conflicting evidence coming from different sources is another reason why one may need to deal with an imperfect model of the system. This can happen, e.g., in case of a distributed knowledge base, where some components may be outdated and/or contain unreliable information. In classical logic, the deductive explosion would make any attempt at reasoning useless. In multi-valued logics, one can assign a special truth value (``inconsistent'') to situations when conflicting evidence about the value of \prop{p} exists, and conduct the verification in the usual way.

Thirdly, even if the model is complete and faithful, the verification procedure can be only partially conclusive.
For instance, in {bounded model checking}~\cite{Biere99bmc,Penczek03ctlk}, a full transition system is given but the formula is checked
on runs of length at most $n$. Now, again, $\Sometm\prop{p}$ is clearly true if we find \prop{p} to occur on every path in up to $n$ steps.
Otherwise, the output is inconclusive (because \prop{p} might or might not occur in subsequent steps).
The case of $\Always\prop{p}$ is analogous.

Note that different sources of multiple truth values can easily combine.
For example, runtime monitoring of a distributed knowledge base may require handling both the possible inconsistency
of data across different locations and the uncertainty about how the system state will evolve.
This kind of scenarios are most conveniently modeled by {partially}, rather than linearly ordered sets of answers.
Overall, it should be up to the designer to decide which domain of truth values will be used for verification.
In this paper, we provide a general methodology, together with flexible algorithms, that can be used for any distributive lattice of interpretation.

\subsection{Technical contribution and outline of the paper}

The focus of the article is on general multi-valued verification of strategic abilities in multi-agent systems.
The main contributions are:
\begin{enumerate}
\item We propose a multi-valued variant of alternating-time temporal logic \ATLs, called \ATLSi, over arbitrary lattices of logical values.
  We also show that the new logic is a conservative extension of the standard, 2-valued \ATLs.
\item\label{it:transl} We study the multi-valued model checking problem for \ATLSi.
Similarly to the previous work on temporal model checking~\cite{kp02a,kp06}, we do not propose dedicated algorithms for \ATLSi.
 Instead, we look for general efficient translations from the multi-valued case to the 2-valued case. In this respect, we:
  \begin{enumerate}
  \item prove that no model-independent translation exists for the whole language of \ATLSi,
  \item identify a broad subclass of formulas for which such a translation can be obtained,
  \item propose a recursive model-dependent translation that works for all instances of the problem.
  \end{enumerate}
\item We show that all results are insensitive to the actual notion of strategy and strategic play. In particular, they easily extend to verification of strategies under imperfect information.
\item Finally, we report an implementation of the verification algorithm based on the translation to 2-valued model checking, and evaluate its performance experimentally.
\end{enumerate}

Point~\ref{it:transl} might require further explanation.
Multi-valued model checking often provides a \emph{conceptual} approximation of the classical (two-valued) verification problem,
especially in case the precise model is difficult to obtain.
On the other hand, we show a \emph{technical} reduction of multi-valued model checking to the two-valued variant.
Thus, typically, a designer who wants to verify properties of a complex system, expressed in (classical) \ATLs,
would first come up with a conceptual approximation of the problem in \ATLSi, then use the technical translation of the specification
back to model checking of \ATLs, and finally run the latter by means of an existing tool (such as MCMAS~\cite{Lomuscio15mcmas}).

The structure of the paper is as follows.
We begin by presenting the context of our work: the alternating time logic \ATLs in Section~\ref{sec:preliminaries},
and distributive lattices in Section~\ref{sec:lattices}. We also introduce our working example of drones monitoring the pollution in a city.
In Section~\ref{sec:multivalued}, we define the syntax and semantics of our multi-valued variant of \ATLs, the logic \ATLSi.
Section~\ref{sec:mcheck-mv} discusses
the model checking problem for \ATLSi (in general and under certain restrictions). Until that point, we assume the classical interpretation
of transitions in models, i.e., the multi-valued approach is applied only to the labeling of \emph{states}.
Section~\ref{sec:mvtrans} extends our definitions and results to models with many-valued transitions.
We also show that the idea of \emph{may/must abstraction} can be seen as a special case of model checking \ATLSi.
Section~\ref{sec:imperfinfo} adapts our results to verification of strategies for agents with imperfect information.
In Section~\ref{sec:drones}, we present an experimental evaluation of our algorithms, based on the drones scenario.
Finally, we conclude in Section~\ref{sec:conclusions}.

\para{Previous version of the article.}
The main concepts and some of the results have appeared in a preliminary form in the conference paper~\cite{Jamroga16mvATL}.
The present version extends it with proper proofs, a comprehensive discussion on motivation and applicability,
and a generalization of the framework to models with multi-valued transitions.
We also add a case study (complete with a detailed experimental evaluation) based on a newly proposed benchmark model.

\subsection{Related work}\label{sec:related}

Multi-valued interpretation of modal formulas has been used in multiple approaches to verification.
The main idea was proposed by Fitting~\cite{fitting91,fitting92} already more than 25 years ago.
Further fundamental work on multi-valued modal logic includes, e.g.,~\cite{Vijzelaar15mv-abstraction}
where general properties of multi-valued abstractions were studied, and demonstrated on an example 9-valued bilattice.

In the 2000s, a number of works adapted the idea to verification of distributed and multi-agent systems.
A variant of \CTLs for models over finite quasi-boolean lattices was proposed in~\cite{kp02a}, together with a general
translation scheme that reduced multi-valued model checking of \CTLs specifications to the standard 2-valued case.
This was later extended to multi-valued modal $\mu$-calculus~\cite{gcconcur03,BrunsG04,ShohamG12,PanLCM15}, and
to multi-valued modal $\mu$-calculus with knowledge~\cite{kp06}.\footnote{
  Note that, despite having ``games'' in the title, \cite{ShohamG12} was not concerned with strategic specifications.
	``Games'' were used there only in the technical sense, to define the semantics of $\mu$-calculus.
However, \cite{AminofKM12} deals with game applications of \cite{ShohamG12}. }
Our paper follows this line of work, and extends the techniques to strategic operators of \ATLs.
We also enrich the language with the two-valued ``implication'' (or comparison) and ''equivalence'' operators $\ra$ and $\equtil$,
which provide:
(i) the notions of material implication and biconditional, useful in specifying general properties
of multi-valued models;
(ii) a way of model checking ``threshold properties'' analogous to probabilistic temporal logics behind PRISM~\cite{Kwiatkowska02prism}.
As it turns out, the new operators require non-trivial treatment, significantly different
from the previous works~\cite{kp02a,gcconcur03,BrunsG04,kp06,ShohamG12,PanLCM15}.

All the above papers consider \emph{general} multi-valued verification, i.e., the interpretation can be based on an arbitrary finite lattice.
Another line of work considers model checking over specific domains of truth values, tailored to a particular class of scenarios.
Model checking methods for the special case of {3}-valued temporal logics were discussed in~\cite{gj03,hjs03,hp03-proc} and,
recently, in~\cite{Kouvaros17predicateAbstraction}.
Two different methods for approximating the standard two-valued semantics of \ATLs under imperfect information
by using three-valued \ATLs are presented and analysed in \cite{BelardinelliLM18,BelardinelliLM19}.
However, these approaches use only Kleene's three valued logic rather than general distributive lattices.
Moreover, a partial algorithm for model checking two-valued perfect recall via its approximation
as three-valued bounded recall is constructed in \cite{BelardinelliLM18}.

Related approaches include runtime verification, which often uses 3-valued~\cite{Bauer06runtime3val}
or 4-valued interpretation~\cite{Bauer07runtime4val,gcconcur03} of temporal formulas.
Moreover, a 4-valued semantics has been used to evaluate database queries~\cite{MaluszynskiS11-4ql,MaluszynskiS13partiality}.
A {3}-valued semantics of strategic abilities in models of perfect information was considered in~\cite{Ball06mv-AMC}
for verification of alternating $\mu$-calculus, and in~\cite{LomuscioM15} for abilities expressed in \ATL.\footnote{
  More precisely, the semantics in~\cite{LomuscioM15} includes agents' indistinguishability relations in the model,
	but it allows for non-uniform play, thus effectively assuming that the agents have perfect information about the current
	state of the system while playing. }
In~\cite{Belardinelli17abstraction}, another 3-valued semantics of \ATL was studied, for both perfect and imperfect information.
In all those papers (i.e.,~\cite{Ball06mv-AMC,LomuscioM15,Belardinelli17abstraction}), the main aim was to verify may/must
abstractions of multi-agent systems.
Note that, while the agenda of our paper comes close to that of~\cite{LomuscioM15,Belardinelli17abstraction},
our semantics differs from~\cite{LomuscioM15} even in the 3-valued case.
Moreover, our multi-valued semantics of \ATL is a conservative extension of standard \ATL, whereas the one in~\cite{LomuscioM15} is not.
In contrast, the \ATL variant in~\cite{Belardinelli17abstraction} \emph{is} a conservative extension of the 2-valued semantics,
and can be in fact considered as a very special case of our general semantics.

A quite different but related strand of research concerns real-valu\-ed logics over probabilistic models for temporal~\cite{Alfaro05discounted,Lluch05quantitative,Jamroga08mtl-aamas} and strategic specifications~\cite{Jamroga08mtl-prima}.
We also mention the research on probabilistic model checking of temporal and strategic
logics~\cite{Huth97quantitative,Baier99computing,Kwiatkowska02prism,Bulling09patl-fundamenta,Huang12probabilisticATL,Chen13prismgames}
that evaluates specifications in the 2-valued domain but recognizes different degrees of success and the need to aggregate
them over available strategies and possible paths.

In sum, there have been approaches to general multi-valued model checking of temporal and epistemic properties, but no analogous proposals for strategic properties.
Furthermore, there were proposals for multi-valued verification of strategic properties over specific (and usually very simple) sets of truth vales,
but no framework that studies the same problem for arbitrary lattices.
This paper fills the gap, and combines elements of both strands to obtain a general framework for multi-valued verification of ability in systems
of interacting agents.

\section{How to specify strategic abilities}\label{sec:preliminaries}

We begin by recalling the basics of two-valued alternating-time temporal logic.
We also introduce our working example that will be used throughout the paper.

\subsection{Syntax}

\emph{Alternating-time temporal logic}~\cite{Alur97ATL,Alur02ATL} generalizes the branching-time temporal
logic \CTLs by replacing path quantifiers $\Epath,\Apath$ with \emph{strategic modalities} $\coop{A}$.
Informally, $\coop{A}\gamma$ says that a group of agents $A$ has a collective
strategy to enforce temporal property $\gamma$.
\ATLs formulas can include temporal operators: $\Next$ (``in the next state''),
$\Always$ (``always from now on''), $\Sometm$ (``now or sometime in the future''),
and $\Until$ (strong ``until'').
Similarly to \CTLs\ and \CTL, we consider two syntactic variants of the alternating-time logic,
namely \ATLs\ and \ATL.
Formally, let $\Agt$ be a finite set of agents, and $\AP$ a countable set of atomic propositions.
The language of \ATLs is defined as follows:
\begin{center}
$\varphi::= \prop{p} \mid \neg \varphi \mid \varphi\wedge\varphi
  \mid \coop{A}\gamma$, \qquad
$\gamma::=\varphi \mid \neg\gamma \mid \gamma\land\gamma \mid
  \Next\gamma \mid \gamma\Until\gamma$.
\end{center}
where $A\subseteq\Agt$ and $\prop{p} \in \AP$.
Traditionally, only the state formulas $\varphi$ are called formulas of \ATLs.
Derived boolean connectives and constants ($\lor,\true,\false$) are defined as usual.
``Sometime'', ``weak until'', and ``always from now on'' are defined as
$\Sometm\gamma \equiv \true \Until \gamma$,
$\gamma_1\WeakUntil\gamma_2 \equiv \neg((\neg\gamma_2)\Until(\neg\gamma_1\land\neg\gamma_2))$,
and $\Always\gamma \equiv \gamma\WeakUntil\false$.
Also, we can use $\noavoid{A}\gamma \equiv \neg \coop{A} \neg\gamma$ to express that, for each strategy of $A$, property $\gamma$ holds on some paths.

\medskip
\ATL (without ``star'') is the syntactic variant in which strategic and temporal operators are combined into compound modalities:
$\varphi ::= \prop{p}  \mid  \lnot \varphi  \mid  \varphi\land\varphi
\mid  \coop{A} \Next\varphi  \mid  \coop{A} \varphi\Until\varphi  \mid \coop{A} \varphi\WeakUntil\varphi$.

\subsection{Models}

The semantics of \ATLs is typically defined over synchronous multi-agent transition systems,
i.e., models where all the agents simultaneously decide on their next actions, and the combination
of their choices determines the next state, see the following definition.

\begin{definition}[CGS]
A \emph{concurrent game structure{ (CGS)}} is a tuple
$M = \tuple{\Agt, \States,
Act, d, \trans, \AP, V}$, which includes  nonempty finite sets of:
agents $\Agt$,
states $\States$,
actions $\Actions$,
atomic propositions $\AP$,
and a propositional valuation $V: \States \rightarrow 2^{\AP}$.
The function $d: \Agt \times \States \rightarrow \powerset{Act}\setminus\set{\emptyset}$ defines the availability of actions.
The (deterministic) transition function
$\trans$ assigns a successor state $q' = \trans(q,\alpha_1,\dots,\alpha_{|\Agt|})$ to
each state $q\in \States$ and any tuple of actions $\alpha_i \in d(i,q)$, one per agent $i\in\Agt$, that
can be executed in $q$.
A \emph{pointed CGS} is a pair $(M,q_0)$, where $M$ is a CGS
and $q_0\in\States$ is the initial state of $M$.
\end{definition}

A \emph{path} $\lambda=q_0q_1q_2\dots$ in a CGS is an infinite sequence of states
such that there is a transition between each $q_i, q_{i+1}$ for each $i \geq 0$.

$\lambda[i]$ denotes the $i$th position on $\lambda$ (starting from $i=0$)
and $\lambda[i,\infty]$ the suffix of $\lambda$ starting with $i$.

\subsection{Semantics}

Given a CGS, we define the strategies and their outcomes as follows.
A \emph{perfect recall strategy} (or \emph{\IR-strategy}) for agent $a$ is
a function $s_a:\States^+ \rightarrow \Actions$ such that $s_a(q_0q_1\dots q_n)\in d(a,q_n)$.
A \emph{memoryless strategy} (or \emph{\Ir-strategy}) for $a$
is a function $s_a:\States\rightarrow \Actions$ such that $s_a(q)\in d(a,q)$.
A \emph{collective strategy} for a group of agents $A=\set{a_1,\ldots,a_r}$
is a tuple of individual strategies $s_A = \tuple{s_{a_1}, \ldots, s_{a_r}}$.
Note that $s_A$ only binds the agents in $A$, while agents outside $A$ can act as they wish.
The set of such strategies is denoted by $\Sigma_A^\IR$ (resp.~$\Sigma_A^\Ir$).
The ``outcome'' function $out(q,s_A)$ returns the set of all paths that
can occur when agents $A$ execute strategy $s_A$ from state $q$ onward.
The semantics of perfect recall \ATLs\ is defined as follows:
\begin{description}
\itemsep=0.7pt
\item $\model,\state \models \prop{p}$ iff $\prop{p}\in\PVal(\state)$, for $\prop{p}\in\AP$;
\item $\model,\state \models \neg\neg\varphi$ iff $\model,\state \models\varphi$, and $\model,\state \models \neg\varphi$ iff $\model,\state \not\models  \varphi$ for $\varphi\neq\neg\psi$ for any $\psi$;
\item $\model,\state \models \varphi_1\wedge\varphi_2$
  iff $\model,\state \models \varphi_1$ and $\model,\state \models \varphi_2$;
\smallskip
\item $\model,q \models \coop{A}\gamma$\quad iff there is a strategy $s_A\in\Sigma_A^\IR$
  such that, for each path $\lambda\in out(q,s_A)$, we have $\model,\lambda \models \gamma$.
\smallskip
\item $\model,\lambda \models \varphi$ iff $\model,\lambda[0]\models \varphi$;
\item $\model,\lambda \models \neg\neg\gamma$ iff $\model,\lambda \models \gamma$, and $\model,\lambda \models \neg\gamma$ iff $\model,\lambda \not\models \gamma$ for $\lambda\neq\neg\gamma$ for any $\gamma$;
\item $\model,\lambda \models \gamma_1\wedge\gamma_2$ iff  $\model,\lambda\models \gamma_1$
  and $\model,\lambda \models \gamma_2$;
\item $\model,\lambda \models \Next\gamma$ iff $\model,\onepath[1,\infty] \models \gamma$; and
\item $\model,\lambda \models \gamma_1\Until\gamma_2$ iff there is an $i\in\mathbb{N}_0$ such that
  $\model,\lambda[i,\infty] \models \gamma_2$ and $\model, \lambda[j,\infty] \models \gamma_1$
  for all $0\leq j< i$,
\end{description}
where $\mathbb{N}_0$ is the set of non-negative integers.
The memoryless semantics of ATL* uses strategies of $\Sigma_A^\Ir$ rather than $\Sigma_A^\IR$.

\begin{figure}[!b]
\vspace{1mm}
\centering
  \includegraphics[scale=0.45]{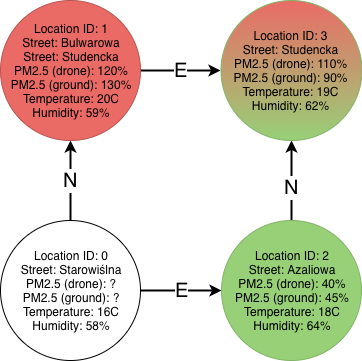}
\caption{Map: drone navigation and measurements in an area of Cracow. Location colors indicate whether the PM2.5 readings  are within or beyond the norm}\label{fig:map}
\end{figure}

\begin{example}[Drones patrolling for pollution]\label{ex:drones-cgs}
Consider a team of $k$ drones monitoring air pollution in the city of Cracow, Poland.
For this scenario, we use the map shown in Figure~\ref{fig:map}.
To keep the resulting CGS small, the map is a grid that only includes four locations, and the drones
can move between the locations in one direction only (either North or East).\footnote{A more complex map
will be used for the experiments in Section~\ref{sec:drones}. }
The initial location is $0$, and the drones are supposed to reach location $3$ (called their ``target'' location).
Each drone has a sensor to measure the level of PM2.5 in the air, that is the rate of particles with a diameter less than 2.5 microns. The drone can also communicate with the nearest sensor
on the ground that reports the PM2.5 rate on the ground level, as well as the current measurements of temperature, air pressure, and humidity.
Note that some measurements may be unobtainable at some locations. In particular, no measurements are available at the start location.
We also assume, for the sake of simplicity, that the environment can be viewed as stationary with respect to the movement of drones,
i.e., the measurements at a location do not change throughout the patrolling mission.

\medskip
Figure~\ref{fig:drones} presents a pointed concurrent game structure $\Mdrones$ that models the above scenario for $k=2$ drones.
Every drone is a separate agent, with two actions available: $N$ (fly North) and $E$ (fly East).
Due to a limited battery capacity, a drone can only visit $l=2$ locations before the battery dies.
Each state of the model includes information about the current locations of the drones, possibly distinguishing
the locations that have been already visited by the team.
Note that in our simple scenario the only pair of locations that can be reached by the drones via two different routes is $(3,3)$.
The corresponding two states are labeled accordingly $(3,3)_1$ and $(3,3)_2$.

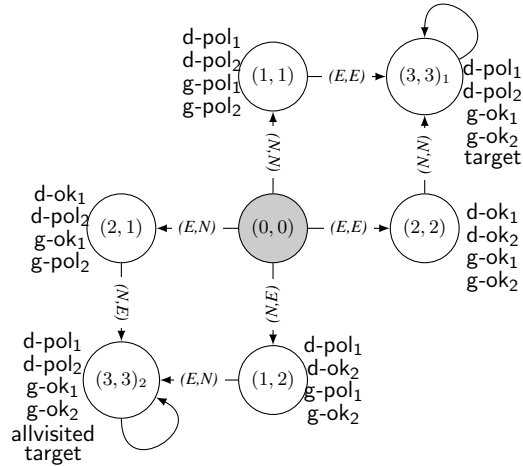
\begin{figure}[!h]
\vspace*{-3mm}
\centering
  \begin{tikzpicture}[>=latex,scale=1.0, every node/.style={scale=0.8}]
    \tikzstyle{state}=[circle,draw,trans, minimum size=8mm]
    \tikzstyle{initstate}=[circle,draw,trans, minimum size=8mm, fill=lightgrey]
    \tikzstyle{trans}=[font=\footnotesize]

    \path (0,0) node[initstate] (q00) {$(0,0)$}
		  (0,2) node[state] (q11) {$(1,1)$}
		  						+(-0.8,0.5) node {\dronepol{1}}
		  						+(-0.8,0.2) node {\dronepol{2}}
		  						+(-0.8,-0.1) node {\groundpol{1}}
		  						+(-0.8,-0.4) node {\groundpol{2}}		  						
		  (2,0) node[state] (q22) {$(2,2)$}
		  						+(0.9,0.2) node {\droneok{1}}
		  						+(0.9,-0.1) node {\droneok{2}}
		  						+(0.9,-0.4) node {\groundok{1}}		
		  						+(0.9,-0.7) node {\groundok{2}}		  						
		  (2,2) node[state] (q33) {$(3,3)_1$}
		  						+(0.9,0.1) node {\dronepol{1}}
		  						+(0.9,-0.2) node {\dronepol{2}}
		  						+(0.9,-0.5) node {\groundok{1}}
		  						+(0.9,-0.8) node {\groundok{2}}
		  						+(0.9,-1.1) node {\prop{target}}		  												  								
		  (-2,-2) node[state] (q33v) {$(3,3)_2$}
		  						+(-0.9,0.5) node {\dronepol{1}}
		  						+(-0.9,0.2) node {\dronepol{2}}
		  						+(-0.9,-0.1) node {\groundok{1}}
		  						+(-0.9,-0.4) node {\groundok{2}}	
		  						+(-0.9,-0.7) node {\prop{allvisited}}		  												
		  						+(-0.9,-1) node {\prop{target}}		  												
		  (0,-2) node[state] (q12) {$(1,2)$}
		  						+(0.8,+0.45) node {\dronepol{1}}
		  						+(0.8,+0.15) node {\droneok{2}}
		  						+(0.8,-0.15) node {\groundpol{1}}
		  						+(0.8,-0.45) node {\groundok{2}}		  									
		  (-2,0) node[state] (q21) {$(2,1)$}
		  						+(-0.8,0.45) node {\droneok{1}}
		  						+(-0.8,0.15) node {\dronepol{2}}
		  						+(-0.8,-0.15) node {\groundok{1}}
		  						+(-0.8,-0.45) node {\groundpol{2}}		  							  	
        ;

   \path[->,font=\scriptsize] (q00)
         edge
           node[midway,sloped]{\onlabel{(N,N)}} (q11)
         edge
           node[midway,sloped]{\onlabel{(E,E)}} (q22)
         edge
           node[midway,sloped]{\onlabel{(N,E)}} (q12)
         edge
           node[midway,sloped]{\onlabel{(E,N)}} (q21);
   \path[->,font=\scriptsize] (q11)
         edge
           node[midway,sloped]{\onlabel{(E,E)}} (q33);
   \path[->,font=\scriptsize] (q22)
         edge
           node[midway,sloped]{\onlabel{(N,N)}} (q33);
   \path[->,font=\scriptsize] (q12)
         edge
           node[midway,sloped]{\onlabel{(E,N)}} (q33v);
   \path[->,font=\scriptsize] (q21)
         edge
           node[midway,sloped]{\onlabel{(N,E)}} (q33v);

   \draw[-latex,black](q33) ..controls +(1.2,0.6) and +(0,1.4).. (q33);
   \draw[-latex,black](q33v) ..controls +(0,-1.4) and +(1.2,-0.6).. (q33v);

  \end{tikzpicture}\vspace*{-3mm}
\caption{Model $\Mdrones$: autonomous drones monitoring pollution\label{fig:drones}}
\end{figure}

Most of the atomic propositions refer to the pollution measurements available to a drone at its current location.
That is, $\dronepol{i}$ indicates that the $i$th drone registers pollution; more precisely:
the sensor of drone $i$ reports that the level of PM2.5 exceeds the norm. Similarly, $\droneok{i}$
indicates that the sensor of drone $i$ reports a level of PM2.5 within the norm.
Moreover, $\groundpol{i}$ (resp.~$\groundok{i}$) says that the ground sensor nearest to drone $i$
reports a level of PM2.5 exceeding the norm (resp.~within the norm).
Proposition \prop{target} indicates states where all drones have reached the target location.
Finally, proposition \prop{allvisited} labels the states where the team has already visited
all locations in the map: in case of $\Mdrones$, the only such state is $(3,3)_2$.
\finis
\end{example}

\begin{example}[Drones, ctd.]\label{ex:drones-atl}
For the model in Example~\ref{ex:drones-cgs}, we have, for instance,
$\Mdrones,(0,0) \models \coop{1}\Sometm\dronepol{1}$: drone $1$ has a strategy ensuring that its sensor will eventually register pollution. The stratety itself is simple: when in the state $(0,0)$ fly North.
Moreover, $\Mdrones,(0,0) \models \coop{1}\Sometm\droneok{1}$: it also has a strategy for reaching a location where it registers no pollution. This time the simplest strategy is to fly East.
In fact, $\Mdrones,(0,0) \models \coop{1}(\Sometm\dronepol{1} \land \Sometm\droneok{1})$: there is a single strategy for achieving both goals. The drone needs to fly East first, and then fly North.

Furthermore, no drone can assure on its own that all of the locations will eventually be visited:
$\Mdrones,(0,0) \models \neg\coop{1}\Sometm\prop{allvisited} \land \neg\coop{2}\Sometm\prop{allvisited}$.
This can only be ensured if both drones cooperate:
$\Mdrones,(0,0) \models  \coop{1,2}\Sometm\prop{allvisited}$.
On the other hand, the drones are bound to end up at the target location, no matter what they decide to do:
$\Mdrones,(0,0) \models  \coop{\emptyset}\Sometm\prop{target}$.
\finis
\end{example}

\section{Multi-valued domains of interpretation}\label{sec:lattices}

As the formulas of our multi-valued version of \ATLs will be interpreted in distributive lattices
of truth values~\cite{Birkhoff48lattice}, we recall the relevant notions and results in this section.

\subsection{Lattices of  truth values}\label{sec:lattices-intro}

\begin{definition}
A {\em lattice} is a partially ordered set $\L = (\LL,\leq)$, where every pair of elements $x,y \in \LL$ has the greatest lower bound
(called {\em meet} and denoted by $x \meet y$) and the least upper bound (called {\em join} and denoted by $x \join y$).
\end{definition}
Note that the meet and join of any $x,y \in \LL$ are uniquely determined due to the antisymmetry of $\leq$.

\medskip
In what follows, we only consider finite lattices.
We denote the least and the greatest elements of $\L$ by $\bot, \top$, respectively.
Also,
we write
(i) $x_1 < x_2 \mbox{ iff } x_1\leq x_2 \mbox{ and } x_1\neq x_2$,
and (ii) $x_1 \bowtie \nobreakspace x_2  \mbox{ iff neither } \\  x_1\leq \nobreakspace x_2 \mbox{ nor } x_2\leq \nobreakspace x_1$.
Moreover, let:
\begin{itemize}
\item $\uparrow x = \{y \in L \mid x \leq y\}$ denote the upward closure of $x$, and
\item $\downarrow x\: = \{y\in L \mid y \leq x\}$ denote the downward closure of $x$.
\end{itemize}
A lattice $\L' = (\LL',\leq')$ is a sublattice of a lattice $\L = (\LL,\leq)$
if $\LL' \subseteq \LL$ and $\leq' = \leq \cap (\LL' \times \LL')$.

\begin{definition}
A lattice $\L=(\LL, \leq)$ is \emph{distributive} if, for any $x,y,z \in \LL$,
the following two conditions hold:
\ (i) $z \join (x \meet y) = (z \join x) \meet (z \join y)$,
\ (ii) $z \meet (x \join y) = (z \meet x) \join (z \meet y)$.
\end{definition}
Note that, for any lattice $\L$, the operations $\join$ and $\meet$ are associative.

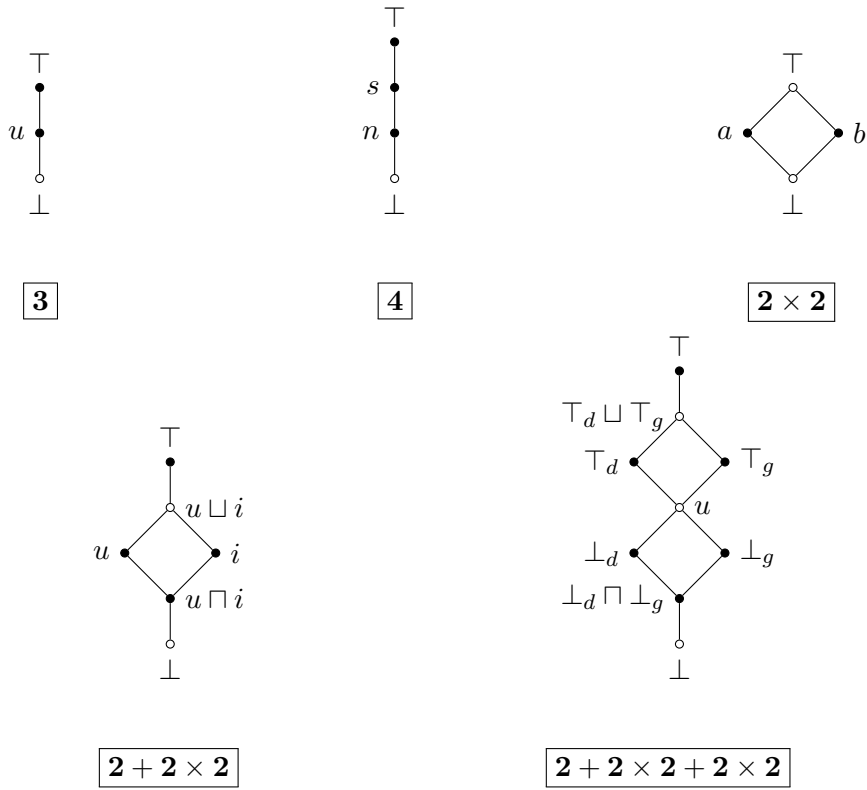
\begin{figure}[!h] 
\centering
\begin{tabular}{c@{\qquad\qquad\qquad\qquad\qquad}c@{\qquad\qquad\qquad\qquad\qquad}c}
\begin{tikzpicture}[scale=0.6]
	\tikzstyle{black}=[draw,circle,fill=black,minimum size=3pt,
                            inner sep=0pt]
    \tikzstyle{white}=[draw,circle,fill=white,minimum size=3pt,
                            inner sep=0pt]

    \draw (0,0) node[black] (t) [label=above:$\top$] {};
    \draw (0,-1) node[black] (u) [label={[align=right]left:$u$}] {};
    \draw (0,-2) node[white] (f) [label=below:$\bot$] {};

    \draw (t) -- (u);
    \draw (u) -- (f);
\end{tikzpicture}
  &
\begin{tikzpicture}[scale=0.6]
	\tikzstyle{black}=[draw,circle,fill=black,minimum size=3pt,
                            inner sep=0pt]
    \tikzstyle{white}=[draw,circle,fill=white,minimum size=3pt,
                            inner sep=0pt]

    \draw (0,0) node[black] (t) [label=above:$\top$] {};
    \draw (0,-1) node[black] (s) [label={[align=right]left:$s$}] {};
    \draw (0,-2) node[black] (n) [label=left:$n$] {};
    \draw (0,-3) node[white] (f) [label=below:$\bot$] {};

    \draw (t) -- (s);
    \draw (s) -- (n);
    \draw (n) -- (f);
\end{tikzpicture}
  &
\begin{tikzpicture}[scale=0.6]
	\tikzstyle{black}=[draw,circle,fill=black,minimum size=3pt,
                            inner sep=0pt]
    \tikzstyle{white}=[draw,circle,fill=white,minimum size=3pt,
                            inner sep=0pt]

    \draw (0,0) node[white] (t) [label=above:$\top$] {};
    \draw (-1,-1) node[black] (a) [label=left:$a$] {};
    \draw (1,-1) node[black] (b) [label=right:$b$] {};
    \draw (0,-2) node[white] (f) [label=below:$\bot$] {};

    \draw (t) -- (a);
    \draw (t) -- (b);
    \draw (a) -- (f);
    \draw (b) -- (f);
\end{tikzpicture}
  \\ \\

\ \ \ \fbox{${\bf 3}$} & \ \ \ \fbox{${\bf 4}$} & \fbox{${\bf 2 \times 2}$}
\end{tabular}

\begin{tabular}{c@{\qquad\qquad\qquad\qquad\qquad}c}
\begin{tikzpicture}[scale = 0.6]
	\tikzstyle{black}=[draw,circle,fill=black,minimum size=3pt,
                            inner sep=0pt]
    \tikzstyle{white}=[draw,circle,fill=white,minimum size=3pt,
                            inner sep=0pt]

    \draw (0,0) node[black] (t) [label=above:$\top$] {};
    \draw (0,-1) node[white] (s) [label=right:$\undec\join\incons$] {};
    \draw (-1,-2) node[black] (dk) [label=left:$\undec$] {};
    \draw (1,-2) node[black] (dc) [label=right:$\incons$] {};
    \draw (0,-3) node[black] (n) [label=right:$\undec\meet\incons$] {};
    \draw (0,-4) node[white] (f) [label=below:$\bot$] {};

    \draw (t) -- (s);
    \draw (s) -- (dk);
    \draw (s) -- (dc);
    \draw (dk) -- (n);
    \draw (dc) -- (n);
    \draw (n) -- (f);
\end{tikzpicture}
  &
\begin{tikzpicture}[scale = 0.6]
	\tikzstyle{black}=[draw,circle,fill=black,minimum size=3pt,
                            inner sep=0pt]
    \tikzstyle{white}=[draw,circle,fill=white,minimum size=3pt,
                            inner sep=0pt]

    \draw (0,0) node[black] (t) [label=above:$\top$] {};
    \draw (0,-1) node[white] (tdutg) [label=left:$\top_d\join\top_g$]{};
    \draw (-1,-2) node[black] (td) [label=left:$\top_d$] {};
    \draw (1,-2) node[black] (tg) [label=right:$\top_g$] {};
    \draw (0,-3) node[white] (u) [label=right:$\undec$] {};
    \draw (-1,-4) node[black] (fd) [label=left:$\bot_d$] {};
    \draw (1,-4) node[black] (fg) [label=right:$\bot_g$] {};
    \draw (0,-5) node[black] (fdnfg) [label=left:$\bot_d\meet\bot_g$]{};
    \draw (0,-6) node[white] (f) [label=below:$\bot$] {};

    \draw (t) -- (tdutg);
    \draw (tdutg) -- (tg);
    \draw (tdutg) -- (td);
    \draw (td) -- (u);
    \draw (tg) -- (u);
    \draw (u) -- (fd);
    \draw (u) -- (fg);
    \draw (fd) -- (fdnfg);
    \draw (fg) -- (fdnfg);
    \draw (fdnfg) -- (f);
\end{tikzpicture}
  \\ \\
\fbox{${\bf 2 + 2 \times 2}$} & \fbox{${\bf 2 + 2 \times 2 + 2 \times 2}$}
\end{tabular}
\caption{Distributive lattices and their join-irreducible elements}\label{fig:jir}\vspace*{-2mm}
\end{figure}

\begin{example}[Some useful lattices]
Figure~3 presents five distributive lattices with applicability motivated by clear practical intuitions.
The total order ${\bf 3}$ is the most popular lattice in multi-valued verification.
The intuition is simple: $\top$ stands for absolute truth, $\bot$ for absolute falsity,
and $u$ can be read as ``unknown'' or ``undefined.'' That is, when a formula $\varphi$ is assigned the value $u$,
this indicates that the statement represented by $\varphi$ cannot be conclusively evaluated
(its truth value -- in the classical sense -- cannot be determined, or even does not exist at the moment).
The lattice is very often used in model checking approaches based on abstraction, with $u$ assigned to formulas
for which the verification has proved inconclusive.

\medskip
The total order ${\bf 4}$ allows for representing graded uncertainty. For instance, in 4-valued approaches to runtime monitoring,
$s$ is interpreted as ``still possibly true,'' and $n$ stands for ``not proved false yet.''
In evidence-based reasoning, the values can correspond to situations when there is much (respectively, little) evidence supporting $\varphi$.
The lattice can be generalized to the $k$-valued linear order ${\bf k}$, useful in scenarios where the amount of positive/negative evidence
is weakly indicative for the truth of a statement. Consider a corpus of data coming from event logs that support or reject proposition \prop{p}.
Then, the logical value of \prop{p} can be, e.g., defined as the difference between the amounts of positive and negative evidence.
A more sophisticated, partially ordered lattice could also involve the  number of conflicts and the support set size.

The partial order ${\bf 2 \times 2}$ can be used to interpret statements with evidence coming from two different, possibly disagreeing sources $A$ and $B$.
Then, the value $a$ can be read as ``true according to source $A$, but not necessarily according to $B$,'' and analogously for $b$.
The dual interpretation is also possible, i.e., we can use $a$ to represent ``false according to source $A$, but not necessarily according to $B$'',
and likewise for $b$.
Actually, these two interpretations correspond to two different choices of the set ${\cal D}$ of so-called {\em designated values},
i.e., values corresponding to truth in classical logic and representing satisfaction of a formula.
Namely, the first interpretation corresponds to ${\cal D} = \{a, \top\}$, and the second --- to ${\cal D} = \{b, \top\}$.
Another useful lattice, ${\bf 2 + 2 \times 2}$, allows for a natural representation of both uncertainty and disagreement.
It provides truth values for statements with inconsistent evidence ($\incons$) and inconclusive evidence ($\undec$).
Combinations of inconsistency and uncertainty can be easily obtained by join and meet ($\undec\join\incons$, $\undec\meet\incons$)

\medskip
For a multi-valued interpretation of the formulas in the drone model, we propose another lattice denoted by $\dlattice$.
The lattice combines two instances of ${\bf 2 \times 2}$: one representing incomplete evidence of truth,
and the other incomplete evidence of falsity.
The idea is that a drone $i$ will evaluate the truth of the proposition $\prop{pol}_i$
(``according to $i$, its current location is polluted''), based on the readings from its own sensor and the nearest ground sensor.
Besides that, the lattice includes explicit nodes for the maximal and minimal elements, similarly to the previous lattice.
This gives us the following basic truth values and their interpretation:
\begin{itemize}
\itemsep=0.9pt
\item $\top$: both readings indicate presence of pollution at the location,
\item $\top_d$: reading from the drone sensor indicates pollution, while the ground sensor indicates no pollution or provides no reading,
\item $\top_g$: reading from the ground sensor shows pollution; absent or negative reading from the drone,
\item $\undec$: there are no readings, neither from the ground nor from the drone,
\item $\bot_d$: drone sensor indicates no pollution; there is no reading from the ground,
\item $\bot_g$: ground sensor indicates no pollution; no reading from the drone,
\item $\bot$: no pollution (both readings are within the norm).
\finis
\end{itemize}
\end{example}

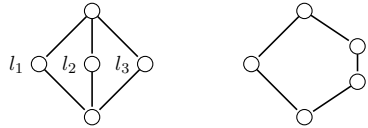
\begin{figure}[!h]
\centering
\begin{tikzpicture}[scale=0.7, transform shape]

  \def\xmv{2.0}
  \def\xmd{1.0}
  
  \node[state, scale=0.3] (s0) at (0.0 -\xmv, 0.0 ) {};
  \node[state, scale=0.3,label=180:{$l_1$}] (s1) at (-\xmd -\xmv, -\xmd) {};
  \node[state, scale=0.3,label=180:{$l_2$}] (s2) at (0.0 -\xmv, -\xmd) {};
  \node[state, scale=0.3,label=180:{$l_3$}] (s3) at (\xmd -\xmv, -\xmd) {};
  \node[state, scale=0.3] (s4) at (0.0 -\xmv, -2*\xmd ) {};
  
  \path [-,style=solid,shorten >=1pt, auto, semithick]
  (s0) edge node {} (s1)
  (s0) edge node {} (s2)
  (s0) edge node {} (s3)
  (s1) edge node {} (s4)
  (s2) edge node {} (s4)
  (s3) edge node {} (s4)
  ;
  
  \node[state, scale=0.3] (n0) at (0.0 +\xmv, 0.0) {};
  \node[state, scale=0.3] (n1) at (-\xmd +\xmv, -\xmd ) {};
  \node[state, scale=0.3] (n2) at (\xmd +\xmv, -2/3 * \xmd ) {};
  \node[state, scale=0.3] (n3) at (\xmd +\xmv, -4/3* \xmd ) {};
  \node[state, scale=0.3] (n4) at (0.0 +\xmv, -2*\xmd ) {};
  \path [-,style=solid,shorten >=1pt, auto, semithick]
  
  (n0) edge node {} (n1)
  (n0) edge node {} (n2)
  (n1) edge node {} (n4)
  (n2) edge node {} (n3)
  (n3) edge node {} (n4)
  ;  
\end{tikzpicture}
\caption{Non-distributive lattices M5 and N5}\label{fig:nondistr}
\end{figure}

Not every lattice is distributive, as shown in Figure~4.
However, distributive lattices have a very simple characterization: a lattice is distributive iff
it contains neither M5 nor N5 as a sublattice~\cite{Birkhoff48lattice}.

\begin{remark}[Quasi-boolean lattices and De Morgan algebras]
The operations of join and meet are natural semantic counterparts of disjunction and conjunction.
Some multi-valued approaches  also add the \emph{complement} operation $\comp$ to the lattice, as the semantic counterpart
of multi-valued negation.
A lattice with complement is usually called \emph{quasi-Boolean}, and when distributive it is referred to as a \emph{De Morgan algebra}.
The most popular case is the lattice underlying 3-valued Kleene logic, used e.g. in \cite{Belardinelli17abstraction, BelardinelliLM19}.
However, the choice of a generic complement to suit any lattice  is problematic from the conceptual point of view.
For instance, what should be the ``opposite'' value to $\incons$ in the lattice ${\bf 2 + 2 \times 2}$? In other words,
if statement $\varphi$ is assessed as ``inconsistent'', what should be the evaluation of ``not $\varphi$?''.
There is no uniform answer to that question,
so in  a general approach paper, like ours, a good strategy is to avoid the negation as much as possible.
This is why instead of negation we have chosen to use here another non-monotonic connective: two-valued implication
representing the lattice order, which is useful in all practical cases  when we want to compare the truth values  of two formulas.
\end{remark}

\subsection{Join-irreducible elements}

\begin{definition}
Let $\L = (\LL,\leq)$ be a lattice.
An element $\ell \in \LL$ is called {\em join-irreducible}
iff $\ell \neq \bot$ and,  for any $x,y \in \LL$,
if $\ell = x \join y$, then either $\ell = x$ or $\ell = y$.
The set of all join-irreducible elements of $\L$
is denoted by $\JI(\L)$.
\end{definition}

It is well known \cite{DP90} that every element $x\neq \bot$ of a finite distributive lattice can
be uniquely decomposed into the join of all join-irreducible elements
in its downward closure, i.e.
\begin{equation}
\label{decomp}
x = \bigjoin (\JI(\L)\; \cap \downarrow x)
\end{equation}

\begin{example}
The join-irreducible elements of the distributive lattices in Figure~\ref{fig:jir}
are marked with black dots.
All other ones can be decomposed into the join of some join-irredu\-cible elements.
\finis
\end{example}

We will use the characterization (\ref{decomp}) to define translations from multi-valued to standard model checking through the following theorem.

\begin{theorem}[\cite{kp06}]
\label{fonJIR}
Let ${\cal L}$ be a finite distributive lattice, and let $\ell \in \JI(\L)$.
Then the \emph{threshold function}\\ $f_\ell: \LL \longrightarrow \{\bot,\top\}$, defined by
$$
f_\ell(x) = \left\{\begin{array}{cl}
           \top & \text{if }x\ge \ell \\
           \bot & \text{otherwise}
           \end{array}\right.
$$
preserves arbitrary bounds, i.e., for an arbitrary set of indices $I$, we have:
\begin{equation}
f_\ell(\bigmeet_{i \in I} x_i) = \bigmeet_{i \in I} f_\ell(x_i)
 \hspace{3em} and \hspace{3em}
f_\ell(\bigjoin_{i \in I} x_i) = \bigjoin_{i \in I} f_\ell(x_i).
\end{equation}
\end{theorem}

\begin{remark}
The above does not hold for lattices which are not distributive.
To see this, consider the element $\ell_1$ of lattice {\bf M5}, which is join-irreducible.
However, $f_{\ell_1}$ does not preserve the upper bounds, as
$f_{\ell_1}(\ell_2 \join \ell_3) = f_{\ell_1}(\top) = \top$ whereas
$f_{\ell_1}(\ell_2) \join f_{\ell_1}(\ell_3) = \bot \join \bot = \bot$.
\end{remark}

\section{Multi-valued strategic logic \ATLSi}\label{sec:multivalued}

In this section we extend the syntax and semantics of \ATLs to allow for multi-valued reasoning.
That is, we propose a variant of \ATLs where formulas are interpreted in an arbitrary lattice $\L = (\LL, \leq)$.

\medskip
It is sometimes useful to refer to logical values in the object language.
To enable this, we assume that a natural interpretation of suitable constants is given, in the following way.

\begin{definition}[Interpreted lattice]
Let $\LConst$ be a countable set of symbols.
An interpreted lattice over $\LConst$ is a triple $\L^+ = (\LL, \leq, \sigma)$, where \mbox{$(\LL, \leq)$} is a lattice,
and $\sigma: \LConst\rightarrow \LL$ is an interpretation of the symbols in $\LConst$ as truth values in $\LL$.
\end{definition}

To explicitly show the connection between the named truth values and their names in $\LConst$,
for any interpreted lattice $\L^+= (\LL, \leq, \sigma)$
and any truth value $x\in \sigma(\LConst)$, we will use the notation $\const{x}$
to denote an arbitrarily selected symbol $c\in \LConst$ such that $\sigma(c)=x$.
We do not make any specific assumptions about $\sigma$. In particular, we do not assume that $\sigma$ must be surjective, as in many situations only some truth values in $\LL$ need to be referred to in formulas. However, in all the examples that follow in this paper, $\sigma$ actually \emph{is} a bijection, since every truth value used there has a specific purpose. Thus, for all our examples, it holds that $\LConst = \{\const{l}\: |\: l\in\LL\}$, and $\sigma(\const{x}) = x$ for all $x\in \sigma(\LConst)$.

\subsection{Syntax}

Since, as explained in Remark 3.4, multi-valued negation can be problematic from the conceptual viewpoint,
we will use instead the binary implication operator $\ra$ corresponding to the lattice order.
Our implication operator is similar to the well-known implication of many-valued Goedel-Dummet logic,
and in general to relevant implication or residuum of lattice meet, of which the former is just a special case.
However, our implication is a two-valued operator, which makes it better suited for proof system purposes.
As it can be used for comparing truth values of formulas, it is of obvious practical importance in many applications,
where we are mainly interested in ascertaining  whether the truth value of a given formula $\varphi$ is greater
than the logical value of some other formula $\psi$ (see Example~\ref{ex:drones-mv} below for more explanations and illustration).
Moreover, in case of two-valued logic, classical negation can be expressed using $\ra$ and the constant representing $\bot.$

To increase the expressive power of the language, we also allow for the use of symbols in $\LConst$.

\medskip
The resulting logic is called \ATLSi\ and has the following syntax:

\begin{center}
$\varphi::= c \mid \prop{p} \mid \varphi\land\varphi \mid \varphi\lor\varphi \mid \varphi\ra\varphi
  \mid \coop{A}\gamma \mid \noavoid{A}\gamma$, \\
$\gamma::=\varphi \mid \gamma\land\gamma \mid \gamma\vee\gamma \mid\Next\gamma \mid \gamma\Until\gamma \mid \gamma\WeakUntil\gamma$.
\end{center}
where $\prop{p}\in \AP$, $A\subseteq\Agt$, and $c\in\LConst$, with $\AP$ being a countable set of atomic propositions,
and $\LConst$ a countable set of constants.

\medskip
In what follows, by an {\em implication formula} we  mean any formula of the form $\varphi_1 \ra\varphi_2$.
Additionally, we define an equivalence formula as $\varphi_1\equtil\varphi_2 = (\varphi_1\ra\varphi_2)\land (\varphi_2\ra\varphi_1)$.

The sublogic of \ATLSi{} without the implication operator will be denoted by \mbox{mv-\ATLs}.

\subsection{Semantics}\label{sec:mv-semantics}

The semantics of \ATLSi\ is defined over concurrent game structures with multi-valued interpretation of atomic propositions.

\begin{definition}[Multi-valued CGS]
Let
$\L^+ = (\LL, \leq, \sigma)$ be an interpreted lattice.
A \emph{multi-valued concurrent game structure (\mvmodel)} over $\L^+$
is a tuple $M = \tuple{\Agt, \States, Act, d, t, \AP, \PVal,\L^+}$,
where $\Agt$, $\States$, $Act$, $d$, $\trans$, $\AP$ are as before, and
$\PVal: \AP \times \States \rightarrow \LL$\ assigns at any state all atomic propositions with truth values from the logical domain $\LL$.
\end{definition}

\begin{example}[Drones ctd.]\label{ex:drones-mv}
A multi-valued model of the drone scenario is presented in Figure~\ref{fig:drones-mv}.
To evaluate atomic propositions and their negations, we use the lattice $\dlattice$ introduced in Section~\ref{sec:lattices-intro}.
Each proposition $\pol{i},\ i=1,\dots,k$, refers to the level of pollution from the viewpoint of drone $i$,
that is, given by the available measurements at the current location of the drone.
Whenever a proposition evaluates to $\bot$, we omit that valuation  from the picture.
\finis
\end{example}

\begin{figure}[!h]
\vspace*{-4mm}
\begin{center}
\hspace{-0.5cm}
\begin{tikzpicture}[>=latex,scale=1.0, every node/.style={scale=0.8}]
  \tikzstyle{state}=[circle,draw,trans, minimum size=8mm]
  \tikzstyle{initstate}=[circle,draw,trans, minimum size=8mm, fill=lightgrey]
  \tikzstyle{trans}=[font=\footnotesize]

  \path (0,0) node[initstate] (q00) {$(0,0)$}
                +(0.8,0.5) node {$[\prop{pol_2}]=\undec$}
                +(0.8,0.85) node {$[\prop{pol_1}]=\undec$}		  								  								  								
    (0,2) node[state] (q11) {$(1,1)$}
                +(0,1.1) node {$[\prop{pol_1}]=\top$}
                +(0,0.7) node {$[\prop{pol_2}]=\top$}		  								  								  								
    (2,0) node[state] (q22) {$(2,2)$}		  							  						
    (2,2) node[state] (q33) {$(3,3)_1$}
                +(1.2,0.1) node {$[\prop{pol_1}]=\top_d$}
                +(1.2,-0.3) node {$[\prop{pol_2}]=\top_d$}	  								
                +(1.2,-0.7) node {$[\prop{target}]=\top$}					
    (-2,-2) node[state] (q33v) {$(3,3)_2$}
                +(0,-0.7) node {$[\prop{pol_1}]=\top_d$}
                +(0,-1.1) node {$[\prop{pol_2}]=\top_d$} 			  												
                +(0,-1.5) node {$[\prop{allvisited}]=\top$}					
                +(0,-1.9) node {$[\prop{target}]=\top$}					
    (0,-2) node[state] (q12) {$(1,2)$}
                +(1.1,0.1) node {$[\prop{pol_1}]=\top $}
    (-2,0) node[state] (q21) {$(2,1)$}
                +(0,0.7) node {$[\prop{pol_2}]=\top$}					  	
      ;

 \path[->,font=\scriptsize] (q00)
       edge
         node[midway,sloped]{\onlabel{(N,N)}} (q11)
       edge
         node[midway,sloped]{\onlabel{(E,E)}} (q22)
       edge
         node[midway,sloped]{\onlabel{(N,E)}} (q12)
       edge
         node[midway,sloped]{\onlabel{(E,N)}} (q21);
 \path[->,font=\scriptsize] (q11)
       edge
         node[midway,sloped]{\onlabel{(E,E)}} (q33);
 \path[->,font=\scriptsize] (q22)
       edge
         node[midway,sloped]{\onlabel{(N,N)}} (q33);
 \path[->,font=\scriptsize] (q12)
       edge
         node[midway,sloped]{\onlabel{(E,N)}} (q33v);
 \path[->,font=\scriptsize] (q21)
       edge
         node[midway,sloped]{\onlabel{(N,E)}} (q33v);

 \draw[-latex,black](q33) ..controls +(1.2,0.6) and +(0,1.4).. (q33);
 \draw[-latex,black](q33v) ..controls +(-1.2,0.6) and +(-1.2,-0.6).. (q33v);

\end{tikzpicture}
\end{center}\vspace*{-6.5mm}
\caption{Multi-valued model $\Mmulti$ for the drone scenario.}
\label{fig:drones-mv}
\end{figure}
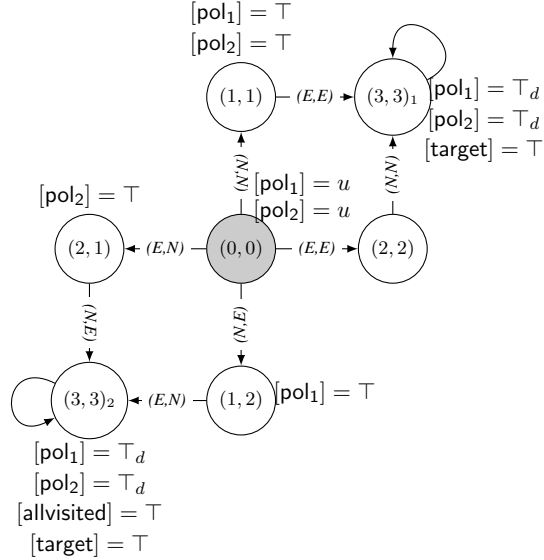

Logical operators can often be naturally interpreted as either maximizers or minimizers
of the truth values. For example, disjunction ($\varphi\lor\psi$) can be understood as a maximizer
(``the most that we can hope to make of either $\varphi$ or $\psi$''), and conjunction as a minimizer
(``the least that we can guarantee for both $\varphi$ and $\psi$'').
This extends to existential quantification (maximizing) and universal quantification
(minimizing) over paths, strategies, moments in time, etc.
Formally, let $M = \tuple{\Agt, \States, Act, d, t, \AP, \PVal,\L^+}$ be an \mvmodel over $\L^+ = (\LL, \leq, \sigma)$.
The valuation function $\truthvalue{\cdot}{}$ is given as below.
We sometimes use $\bigmeet_X\set{Y}$ as a shorthand for $\bigmeet\set{Y \mid X}$,
and similarly for the supremum.
For any $q \in \States$ and any path $\lambda$ in $M$, we define:

\begin{description}
\itemsep=0.9pt
\item $\truthvalue{c}{M,q} = \sigma(c)$\quad for $c\in \LConst$;
\item $\truthvalue{\prop{p}}{M,q} = \PVal(\prop{p},q)$
  \quad for $\prop{p}\in \AP$;

\smallskip
\item $\truthvalue{\varphi_1\land\varphi_2}{M,q} =
  \truthvalue{\varphi_1}{M,q} \meet \truthvalue{\varphi_2}{M,q}$;
\item $\truthvalue{\varphi_1\lor\varphi_2}{M,q} =
  \truthvalue{\varphi_1}{M,q} \join \truthvalue{\varphi_2}{M,q}$;
\item $\truthvalue{\gamma_1\land\gamma_2}{M,\lambda} = \truthvalue{\gamma_1}{M,\lambda} \meet \truthvalue{\gamma_2}{M,\lambda}$\quad
  and\quad $\truthvalue{\gamma_1\lor\gamma_2}{M,\lambda} = \truthvalue{\gamma_1}{M,\lambda} \join \truthvalue{\gamma_2}{M,\lambda}$;

\smallskip
\item $\truthvalue{\varphi}{M,\lambda} =
  \truthvalue{\varphi}{M,\lambda[0]}$;

\smallskip
\item $\truthvalue{\Next\gamma}{M,\lambda} =
  \truthvalue{\gamma}{M,\lambda[1..\infty]}$;

\item $\truthvalue{\gamma_1\Until\gamma_2}{M,\lambda} =
  \bigjoin_{i \in \mathbb{N}_0}\bigmeet_{0\le j<i} \set{\truthvalue{\gamma_2}{M,\lambda[i..\infty]}\meet \truthvalue{\gamma_1}{M,\lambda[j..\infty]}}$;

\item $\truthvalue{\gamma_1\WeakUntil\gamma_2}{M,\lambda} =
  \bigmeet_{i \in \mathbb{N}_0} \set{\truthvalue{\gamma_1}{M,\lambda[i..\infty]}} \join\
  \bigjoin_{i\in \mathbb{N}_0}\bigmeet_{0\le j<i} \set{\truthvalue{\gamma_2}{M,\lambda[i..\infty]}\meet \truthvalue{\gamma_1}{M,\lambda[j..\infty]}}$;

\medskip
\item $\truthvalue{\coop{A}\gamma}{M,q} =
  \bigjoin_{s_A\in\Sigma_A}\bigmeet_{\lambda\in out(q,s_A)}\set{\truthvalue{\gamma}{M,\lambda}}$;
\item $\truthvalue{\noavoid{A}\gamma}{M,q} =
  \bigmeet_{s_A\in\Sigma_A}\bigjoin_{\lambda\in out(q,s_A)}\set{\truthvalue{\gamma}{M,\lambda}}$;

\smallskip
\item $\truthvalue{\varphi_1\ra\varphi_2}{M,q} =
  \trueval$ if $\truthvalue{\varphi_1}{M,q}\le\truthvalue{\varphi_2}{M,q}$
  and $\falseval$ otherwise.
\end{description}
It is worth noting that our implication operator differs from the well-known residue of lattice meet in being
two-valued --- which makes it better suited for use in any proof system, and more intuitive in specification of many requirements.

\medskip
The semantics of the two ``until'' operators demands a more detailed explanation.
The computation of $\truthvalue{\gamma_1\Until\gamma_2}{M,\lambda}$ seeks to achieve a position $i$ on path $\lambda$,
for which the value of $\gamma_2$ at $\lambda[i]$, and the values of $\gamma_1$ at all the points preceding $\lambda[i]$, are guaranteed maximal.
The semantics of $\gamma_1\WeakUntil\gamma_2$ is based on the well-known unfolding $\gamma_1\WeakUntil\gamma_2 \equiv (\Always\gamma_1) \lor (\gamma_1\Until\gamma_2)$, transformed here to a multi-valued interpretation.
Note also that in case of the derived temporal operators ``sometime'' and ``always'' the semantic rules reduce to:
\begin{description}
\itemsep=0.9pt
\item $\truthvalue{\Sometm\gamma}{M,\lambda} =
  \bigjoin_{i\in\Nat}\truthvalue{\gamma}{M,\lambda[i..\infty]}$;

\smallskip
\item $\truthvalue{\Always\gamma}{M,\lambda} =
  \bigmeet_{i\in\Nat}\truthvalue{\gamma}{M,\lambda[i..\infty]}$.
\end{description}

Thus, for instance, the formula $\coop{A}\Sometm\prop{pol}$ can be read as:
``the maximal level of pollution readings that $A$ can guarantee to reach.''
Clearly, such statements do not always submit to intuitive understanding, in particular when nested strategic operators are used.
Because of that, we will stick to simple formulas in our working examples, that is, ones that are relatively easy to read.

\begin{example}[Drones ctd.]\label{ex:drones-mvatl}
For the model in Figure~\ref{fig:drones-mv}, we have
$\truthvalue{\coop{1}\Sometm\pol{1}}{\Mmulti,(0,0)} = \top$: there is a strategy for drone 1
to surely detect pollution (the strategy being to fly North in state $(0,0)$, and then East in $(1,1)$ or $(1,2)$).
Similarly for the other drone we have $\truthvalue{\coop{2}\Sometm\pol{2}}{\Mmulti,(0,0)} = \top$ (the same strategy, but now executed by drone 2).
On the other hand, $\truthvalue{\coop{1}\Always\pol{1}}{\Mmulti,(0,0)} = \undec$: the maximal \emph{guaranteed}
level of detection throughout the mission is $\undec$ (obtained by the same strategy again).
This means that if drone 1 wants to maximize its detection level, the best it can achieve is to keep
it consistently at the level of ``uncertain'' or higher.
Finally, \\
\centerline{$\truthvalue{\coop{1,2}\Sometm(\prop{target} \land \prop{allvisited} \land (\pol{1}\lor\pol{2}))}{\Mmulti,(0,0)} = \top_d$.}\\
That is, if the drones cooperate, and their goal is to reach the target, visit all the locations on the way, and at the end get the pollution detected by at least one of them, then their degree of success is $\top_d$ (pollution indicated by the drone sensor but not by the ground sensor).
\finis
\end{example}

The logical constants we have introduced are especially useful in implication formulas, as the subsequent example demonstrates.

\begin{example}[Implication formulas]\label{ex:comparing}
The ``implication''  operator provides several interesting specification patterns.
For instance, it allows for specifications that are accepted when the ``strength''
of a property reaches a given threshold, similarly to the probabilistic approaches of~\cite{Kwiatkowska02prism,Chen13prismgames}.
As an example, the formula $\const{\undec} \ra \coop{1}\Always\pol{1}$ can be used to specify that the truth value of $\coop{1}\Always\pol{1}$ is at least $\undec$ (intuitively: there is no evidence that the formula is false).
It is easy to see that the formula is true in the model of Figure~\ref{fig:drones-mv}; formally:
$\truthvalue{\const{\undec} \ra \coop{1}\Always\pol{1}}{\Mmulti,(0,0)} = \top$.
Naturally, any stronger requirement on the value of $\coop{1}\Always\pol{1}$ evaluates to ``false,'' e.g.,
$\truthvalue{\const{\top} \ra \coop{1}\Always\pol{1}}{\Mmulti,(0,0)} = \bot$.

Moreover, the formula $\coop{1}\Sometm\pol{1} \ra \coop{2}\Sometm\pol{2}$ says that the ability of drone 2 to spot pollution is at least as good as that of drone 2 (the formula evaluates to $\top$ in $\Mmulti,(0,0)$).
Finally, $\coop{1}\Sometm(\pol{1}\equtil\const{\top_g})$ says that the first drone has a strategy to ensure that it will reach a location where only the ground sensor indicates pollution.
Clearly, the last formula evaluates to $\bot$ in $\Mmulti,(0,0)$.
\finis
\end{example}

We note that most approaches to general multi-valued model checking of temporal
specifications~\cite{kp02a,gcconcur03,kp06,ShohamG12} allow also for \emph{multi-valued transitions} in the models,
analogous to probabilistic transitions in Markov chains and Markov Decision Processes. That is, transitions can be
assigned ``weights'' drawn from the same algebra $\L$.
Similarly, most 3-valued approaches to temporal abstraction and model checking implicitly assume 3-valued transitions
by distinguishing between \emph{may} and \emph{must}
transitions~\cite{Godefroid02abstraction,gj03,hjs03,hp03-proc}.
However, the two approaches differ in how such transitions affect the semantics of formulas with universal
quantification (such as ``for all paths $\gamma$''). In the general multi-valued approach,
the ``weaker'' the path is, the more it decreases the value of the formula.
In the 3-valued approach, ``weaker'' paths have less influence on the overall value.
We do not engage in this discussion here, and leave a proper treatment of multi-valued transitions until Section~\ref{sec:mvtrans}.

\subsection{Truth Levels}

We assume that $\top$ is a single designated value, standing for full logical truth.
In consequence, the truth and validity of formulas can be defined in a straightforward
way as follows:

\begin{definition}[Validity levels]
\label{def:validity}

Let $M$ be \mvmodel, $q$ a state in $M$, and $\varphi$ a state formula of \ATLSi.
Then:
\begin{itemize}
\item $\varphi$ is \emph{true
  in $M,q$} (written $M,q\models\varphi$) iff $\truthvalue{\varphi}{M,q} = \top$.
\item $\varphi$ is \emph{valid in $M$} (written $M\models\varphi$) iff
  $\varphi$ is true in every state of $M$.
\item $\varphi$ is \emph{valid} (written $\models\varphi$)
  iff $\varphi$ is valid in every \mvmodel\ $M$.
\item Additionally, for a path formula $\gamma$, we can say that
  $\gamma$ holds on run $\lambda$ in a \mvmodel\ $M$
  (written $M,\lambda\models\gamma$) iff $\truthvalue{\gamma}{M,\lambda} = \top$.
\end{itemize}
\end{definition}

We now show that \ATLSi{} agrees with standard \ATLs on 2-valued models, unlike the 3-valued version
of \ATLs from~\cite{LomuscioM15}.

\begin{theorem}
\label{conservative}
The logic \ATLSi{} is a conservative extension of
\ATLs, i.e., every CGS $M$ for \ATLs can be identified with an mv-CGS $M'$ for \ATLSi{} over the lattice {\bf 2} such that,
for any \ATLs formula $\varphi$ and any state (path) $\xi$, we have
$M',\xi\models\varphi \mbox{ iff } M,\xi\models\varphi$.
\end{theorem}

\begin{proof}
For any CGS $M = \tuple{\Agt, \States,  Act, d, \trans, \AP, \PVal}$ for \ATLs,
let $M' = \tuple{\Agt, \States, Act, d, \trans, \AP, \PVal',{\bf 2}}$,
where: (i) ${\bf 2} = (\{\bot,\top\}, \leq, \sigma)$ is an interpreted classical lattice of two truth values over
${\cal C} = \{\const{\bot}, \const{\top}\}$ and $\sigma (\const{l}) = l$ for any $l\in\{\bot, \top\}$;
(ii) $\PVal'(p,q)= \top$ if $q \in \PVal(p)$ and $\bot$ otherwise.
Then $M'$ is an \mvmodel\ for \ATLSi,
and an easy check shows that, for any
\ATLs formula $\varphi$ and  any state (path) $\xi$, it indeed obtains $M',\xi\models\varphi$ iff $M,\xi\models\varphi$.
\end{proof}

\section{Model checking \ATLSi}
\label{sec:mcheck-mv}

Given an \mvmodel\ $M$, a state $q$ in $M$, and an \ATLSi formula $\varphi$,
the model checking problem consists in computing the value of $\truthvalue{\varphi}{M,q}$.
This can be done in two ways:
either by using a dedicated algorithm, or through an efficient reduction to the "classical", 2-valued version of model checking.
The latter option has many advantages.
First and foremost, it allows us to benefit from the ongoing developments in 2-valued model checking,
including symbolic model checking techniques, heuristics, model reduction techniques, etc.
In this section, we show how model checking of \ATLSi can be reduced to the 2-valued variant of this problem.
Since a basic result underlying such reduction holds for distributive lattices only, throughout the section
we assume that all lattices under consideration are distributive, unless stated to the contrary.

We emphasize again that, while multi-valued model checking typically provides a \emph{conceptual} approximation
of classical verification, the results in this section are about something else.
Here, we look for a \emph{technical} reduction from multi-valued to two-valued model checking, with the sole purpose
of facilitating the verification process.

\subsection{From multi-valued model checking to classical model checking}\label{sec:mcheck-mv-translation}

It is well known that model checking multi-valued temporal logics can be reduced to classical,
2-valued model checking~\cite{kp02a,gcconcur03,BrunsG04,kp06}. The reduction is of one-to-many type,
i.e., a single instance of multi-valued model checking translates to linearly many instances of classical model checking.
The key result in this respect is~\cite[Theorem 1]{kp02a}.  It proposes a method for ``clustering'' the truth values
from lattice $\L$ into a smaller lattice $\L'$ in such a way that the outcome of model checking is preserved.
We will now show that the analogue of that theorem holds for \mvATL, i.e., the sublanguage of \ATLSi without the $\ra$ operator.

\begin{definition}
\begin{enumerate}
\item By a lattice reduction triple (LRT) we mean a triple $(\L, \L_f, f)$, where $\L = (\LL, \leq)$ is an arbitrary finite
lattice, $\L_f = (\LL_f, \leq_f)$ its sublattice, and $f: \LL \rightarrow \LL_f$  a homomorphism --- a mapping which preserves
arbitrary bounds in $\L$, i.e. such that
\begin{equation}
\label{bounds}
f(\bigmeet_{i \in I} x_i) = \bigmeet_{i \in I} f(x_i)
 \hspace{3em} and \hspace{3em}
f(\bigjoin_{i \in I} x_i) = \bigjoin_{i \in I} f(x_i)
\end{equation}
for an arbitrary set of indices $I$.
\item
Given an LRT $(\L, \L_f, f)$ and an \mvmodel{} $M = $ $\tuple{\Agt, \States, Act, d, \trans, \AP, \V, \L^+}$
over an interpreted lattice $\L^+ = (\LL, \leq, \sigma)$, by the {\em reduction of $M$ to $\L_f$ via $f$} we mean the \mvmodel
$f(M) = \tuple{\Agt, \States, Act, d, \trans, \AP, \V_f, (\LL_f, \leq_f, \sigma_f)}$,
where
\begin{enumerate}
\item $\sigma_f(c)= f(\sigma(c))$\  for any $c\in \LConst$, and
\item $\V_f(p,q) = f(\V(p,q))$\ for any $q \in \States$ and $p \in \AP$.
\end{enumerate}
\end{enumerate}
\end{definition}

\begin{definition}
For any LRT $(\L, \L_f, f)$ and any model $M$ over $\LL$, by the \em{translation condition} for LRT and formula $\vfi$ we mean the relationship
\begin{equation}
\label{thesis1}
\truthvalue{\varphi}{M,\xi} \in f^{-1}(x)\qquad \mbox{iff}\qquad \truthvalue{\varphi}{f(M),\xi} = x
\end{equation}
holding for any state (respectively, path) $\xi$.
\end{definition}
The proof of the theorem follows easily from the key result given below:
\begin{lemma}\label{prop:lem-gen}
Let a state or path formula $\vfi$ be such that
$$\truthvalue{\vfi}{M,\xi} = \bigjoin_{i\in I} \bigmeet_{j_i \in J_i}\truthvalue{\vfi_{j_i}}{M,\xi_{j_i}}\;\;\;
\mbox{or} \;\;\;\;
\truthvalue{\vfi}{M,\xi} = \bigmeet_{i\in I}\bigjoin_{j_i \in J_i} \truthvalue{\vfi_{j_i}}{M,\xi_{j_i}}
$$
for any \mvmodel\ $M$, any states and/or paths $\xi, \xi_{j_i}$ of $M$,
any countable sets $I,J_i$,
and state (resp. path) formulas of \mvATL $\vfi_{j_i}$ for $j_i \in J_i,i \in I$, such that all $\vfi_{j_i}$'s satisfy translation condition (\ref{thesis1}).
Then $\vfi$  satisfies the translation condition too.
\end{lemma}

\begin{proof}
We consider the case
$\truthvalue{\vfi}{M,\s} = \bigjoin_{i\in I} \bigmeet_{j_i \in J_i}\truthvalue{\vfi_{j_i}}{M,\s_{j_i}}$;
the other case follows by symmetry.
As $f$ preserves the bounds, by the assumption on $\vfi$ we have
$f(\truthvalue{\vfi}{M,\s})  = \bigjoin_{i\in I} \bigmeet_{j_i \in J_i}f(\truthvalue{\vfi_{j_i}}{M,\s_{j_i}}\!)$.
Each $\vfi_{j_i}$ satisfies (\ref{thesis1}), so
$f(\truthvalue{\vfi}{M,\s}) = \bigjoin_{i\in I} \bigmeet_{j_i \in J_i}\truthvalue{\vfi_{j_i}}{M_f,\s_{j_i}} =
\truthvalue{\vfi}{M_f,\s}$,
whence
$\truthvalue{\vfi}{M_f,\s} = x$ iff
$\truthvalue{\vfi}{M,\s} \in f^{-1}(x)$, and (\ref{thesis1}) holds for $\vfi$.
\end{proof}
Now, we can formulate the reduction theorem.
\begin{theorem}[Reduction theorem]
\label{general}
Let $\L = (\LL, \leq)$ be an arbitrary finite lattice, and $(\L, \L_f, f)$ an LRT.
Further, let $M$ be an \mvmodel over an interpreted lattice $\L^+ = (\LL, \leq, \sigma)$
with $M=\tuple{\Agt, \States, Act, d, \trans, \AP, \V, \L^+}$ ,
and let $f(M) = \tuple{\Agt, \States, Act, d, \trans, \AP, \V_f, (\LL_f, \leq_f, \sigma_f)}$ be the image of $M$ under $f$.
Then, for any state (respectively, path) formula $\varphi$ of \mvATL over $\L$ and any state (respectively, path) $\xi$,
the following  condition is satisfied:
\begin{equation}
\label{thesis}
\truthvalue{\varphi}{M,\xi} \in f^{-1}(x)\qquad \mbox{iff}\qquad \truthvalue{\varphi}{f(M),\xi} = x
\end{equation}
\end{theorem}

\begin{proof}
We use induction on the length of a formula.
Equation (\ref{thesis}) clearly holds for propositional variables and
their negations. Assume it holds for formulas of length at most $k$, and consider formula $\varphi$ of length $k+1$.
Then we have the following cases:
\begin{enumerate}
\item [(a)]
$\varphi = \varphi_1 \wedge \varphi_2$ or $\varphi = \varphi_1 \vee \varphi_2$,
where each $ \varphi_i$ is a formula of length at most $k$.
Then $\truthvalue{\varphi_1}{M,q}=\truthvalue{\varphi_1}{M,q} \meet \truthvalue{\varphi_2}{M,q}$ or
$\truthvalue{\varphi_1}{M,q}=\truthvalue{\varphi_1}{M,q} \join \truthvalue{\varphi_2}{M,q}$, respectively, where
$\varphi_1, \varphi_2$ satisfy (\ref{thesis}).
As  $\truthvalue{\varphi_1}{M,q}$ is in one of the two dual forms prescribed by Lemma~\ref{prop:lem-gen}
for $I=\{1\}$ and $J_1=\{1,2\}$, by that lemma, $\varphi$ must satisfy (\ref{thesis}), too.
\item[(b)]
$\gamma = \gamma_1 \wedge \gamma_2$ or $\gamma = \gamma_1 \vee \gamma_2$ --- analogously to (a).

\item[(c)] $\varphi= {\rm X} \psi$, where $\psi$ is of length at most $k$. Then $\truthvalue{\Next\gamma}{M,\lambda} =
  \truthvalue{\gamma}{M,\lambda[1..\infty]}$, and as $\psi$ satisfies (\ref{thesis}) by inductive hypothesis, so obviously does $\varphi$.
The reasoning is similar to (a).

\item [(d)]
$\varphi = {\rm U} \psi$, where $\truthvalue{\gamma_1\Until\gamma_2}{M,\lambda} =
  \bigjoin_{i \in \mathbb{N}_0}\bigmeet_{0\le j<i} \set{\truthvalue{\gamma_2}{M,\lambda[i..\infty]}\meet \truthvalue{\gamma_1}{M,\lambda[j..\infty]}}$;
	
Since the operator ${\rm U}$ corresponds to a combination of finite and infinite lower and upper bounds applied to values of formulas of length
at most $k$ for which (\ref{thesis}) holds,
then by Lemma~\ref{prop:lem-gen} Equation (\ref{thesis}) must hold for $\varphi$ too.

\item [(e)]
$\varphi = {\rm W} \psi$, where \\
\centerline{
$\truthvalue{\gamma_1\WeakUntil\gamma_2}{M,\lambda} =
  \bigmeet_{i \in \mathbb{N}_0} \set{\truthvalue{\gamma_1}{M,\lambda[i..\infty]}} \join
  \bigjoin_{i\in \mathbb{N}_0}\bigmeet_{0\le j<i} \set{\truthvalue{\gamma_2}{M,\lambda[i..\infty]}\meet \truthvalue{\gamma_1}{M,\lambda[j..\infty]}}$;
}

Since the operator ${\rm W}$ corresponds to a combination of finite and infinite lower and upper bounds applied to values of formulas of length
at most $k$ for which (\ref{thesis}) holds,
then by Lemma~\ref{prop:lem-gen} Equation (\ref{thesis}) must hold for $\varphi$ too.

\item [(f)] $\varphi= \coop{A}\gamma$, where $\ga$ is of length at most $k$. Then $\ga$ satisfies (\ref{thesis}) by inductive hypothesis,
  and as $\truthvalue{\coop{A}\gamma}{M,q} =
  \bigjoin_{s_A\in\Sigma_A}\bigmeet_{\lambda\in out(q,s_A)}\truthvalue{\gamma}{M,\lambda}$,
	then $\varphi$ satisfies (\ref{thesis}) by Lemma~\ref{prop:lem-gen}.

\item [(g)] $\varphi=\noavoid{A}\gamma$ --- analogously to (e).
\end{enumerate}

\vspace*{-6mm}
\end{proof}

Note that the mapping $f$ can be seen as an abstraction of truth values similar to the well-known technique
of \emph{state abstraction}~\cite{Cousot77abstraction,Clarke94abstraction}.
That is, we can view each value $x\in \LL_f$ as an \emph{abstract truth value} corresponding to the subset $f^{-1}(x)$
of the original truth values in $\LL$.
Clearly, those subsets partition $\LL$ into equivalence classes.
Theorem~\ref{general} says that if $f$ satisfies conditions (\ref{bounds}), then model checking in the abstract model
$M_f$ yields the equivalence class corresponding to the output of the original model checking problem in the concrete model $M$.

\medskip
How can we use Theorem~\ref{general} to reduce multi-valued model checking to the 2-valued case?
Recall the threshold functions $f_\ell: \LL \longrightarrow \{\bot,\top\}$, defined by
$$
f_\ell(x) = \left\{\begin{array}{cl}
           \top & \text{if }x\ge \ell \\
           \bot & \text{otherwise}
           \end{array}\right.
$$
We already stated in Theorem~\ref{fonJIR} that those functions preserve bounds. The following is an immediate corollary of the above:

\begin{corollary}
\label{prop:twovalued}
For any state (respectively, path) formula $\varphi$ of \mvATL and any state (respectively, path) $\xi$, we have
\begin{equation}
\truthvalue{\varphi}{M,\xi} \ge \ell\qquad \mbox{iff}\qquad M_{f_\ell},\xi\models\varphi.
\end{equation}
\end{corollary}

Note that each $M_{f_\ell}$ is a classical, 2-valued model.
Together with Equation (\ref{decomp}), we have
$\truthvalue{\varphi}{M,\xi} = \bigjoin\{ \ell \in \JI(\L) \;\mid \; \truthvalue{\varphi}{M^\ell,\xi} = \top\}$.
This gives us a simple algorithm for computing $\truthvalue{\varphi}{M,\xi}$,
presented in Figure~\ref{fig:mcheck-translation}.
The following is straightforward.

\begin{figure}[h]\small
\vspace{2mm}
\centering
\begin{myalgorithm}{$mcheck_{tr}(M,\xi,\varphi)$}\vspace*{-2mm}
\itemsep=0.4pt
\item For every join-irreducible logical value $\ell\in\JI(\L)$, model-check (classically) $\varphi$ in ${M_{f_\ell},\xi}$;
\item Collect the values of $\ell$ for which the answer was ``yes'' in a set $\mathcal{X}$;
\item Return the join of the values in $\mathcal{X}$, i.e., $\bigjoin \mathcal{X}$.
\end{myalgorithm}
\caption{Translation-based model checking for \mvATL} \label{fig:mcheck-translation}\vspace*{-3mm}
\end{figure}

\begin{theorem}
The one-to-many reduction from multi-valued model checking of \mvATL to 2-valued model checking
of \ATLs runs in linear time with respect to the size of the model and the number of truth values.
\end{theorem}

\begin{example}[Testing of the drones]
Consider the pollution monitoring scenario from the previous examples.
Suppose that we want to test the design of a drone patrol before its deployment in the physical environment.
One way to carry out offline testing is to model-check the relevant properties of the design against a randomly generated
sample of area maps. For the clarity of the examples we have used a crafted by hand map.
For the map in Figure~\ref{fig:map} and the \mvmodel $\Mmulti$ in Figure~\ref{fig:drones-mv},
we obtain the collection of classical models presented in Figure~\ref{fig:drones-translation2}.
Note that projections $(\Mmulti)_{f_{\top}}$ and $(\Mmulti)_{f_{\top\!_g}}$ are in fact identical,
and similarly for $(\Mmulti)_{f_{\bot_d}}$, $(\Mmulti)_{f_{\bot_g}}$, and $(\Mmulti)_{f_{\bot_d\meet\bot_g}}$.

\medskip
Suppose now that we want to compute the value of $\coop{1,2}\Sometm(\prop{target} \land \prop{allvisited} \land (\pol{1}\lor\pol{2}))$ in $\Mmulti,(0,0)$.
The formula holds in state $(0,0)$ of models $(\Mmulti)_{f_{\top\!_d}}$, $(\Mmulti)_{f_{\bot_d}}$, $(\Mmulti)_{f_{\bot_g}}$, and $(\Mmulti)_{f_{\bot_d\meet\bot_g}}$, but not in $(\Mmulti)_{f_{\top}}$ and $(\Mmulti)_{f_{\top\!_g}}$.
Thus, the output of model checking is $\top_d\,\join\,\bot_d\,\join\,\bot_g\,\join\,(\bot_d\meet\bot_g) = \top_d$.
Moreover, to model-check $\coop{1}\Sometm\pol{1}$, we observe that the formula holds in all the projection
models in Figure~\ref{fig:drones-translation2}.
Thus, its value in $\Mmulti,(0,0)$ is $\top\,\join\,\top_d\,\join\,\top_g\,\join\,\bot_d\,\join\,\bot_g\,\join\,(\bot_d\meet\bot_g) = \top$.
\finis
\end{example}

\begin{figure}[!ht]
\vspace*{-2.5mm}
\centering
	\begin{tabular}{@{}c@{\qquad}c@{}}
  \begin{tabular}{@{}c@{}}
  \scalebox{0.9}{
  \begin{tikzpicture}[>=latex,scale=0.8, every node/.style={scale=0.7}]
    \tikzstyle{state}=[circle,draw,trans, minimum size=8mm]
    \tikzstyle{initstate}=[circle,draw,trans, minimum size=8mm, fill=lightgrey]
    \tikzstyle{trans}=[font=\footnotesize]

    \path (0,0) node[initstate] (q00) {$(0,0)$}
          (0,2) node[state] (q11) {$(1,1)$}
		  						+(0,1.1) node {$\prop{pol_1}$}
		  						+(0,0.7) node {$\prop{pol_2}$}
          (2,0) node[state] (q22) {$(2,2)$}	
		  (2,2) node[state] (q33) {$(3,3)_1$}
		  						+(1,0) node {$\prop{target}$}
		  (-2,-2) node[state] (q33v) {$(3,3)_2$}
		  						+(0,-0.8) node {$\prop{allvisited}$}					
		  						+(0,-1.2) node {$\prop{target}$}					
		  (0,-2) node[state] (q12) {$(1,2)$}
		  						+(0.9,0.1) node {$\prop{pol_1}$}
		  (-2,0) node[state] (q21) {$(2,1)$}
		  						+(0,0.7) node {$\prop{pol_2}$}					  	
        ;

   \path[->,font=\scriptsize] (q00)
         edge
           node[midway,sloped]{\onlabel{(N,N)}} (q11)
         edge
           node[midway,sloped]{\onlabel{(E,E)}} (q22)
         edge
           node[midway,sloped]{\onlabel{(N,E)}} (q12)
         edge
           node[midway,sloped]{\onlabel{(E,N)}} (q21);
   \path[->,font=\scriptsize] (q11)
         edge
           node[midway,sloped]{\onlabel{(E,E)}} (q33);
   \path[->,font=\scriptsize] (q22)
         edge
           node[midway,sloped]{\onlabel{(N,N)}} (q33);
   \path[->,font=\scriptsize] (q12)
         edge
           node[midway,sloped]{\onlabel{(E,N)}} (q33v);
   \path[->,font=\scriptsize] (q21)
         edge
           node[midway,sloped]{\onlabel{(N,E)}} (q33v);

   \draw[-latex,black](q33) ..controls +(1.2,0.6) and +(0,1.4).. (q33);
   \draw[-latex,black](q33v) ..controls +(-1.2,0.6) and +(-1.2,-0.6).. (q33v);

  \end{tikzpicture}
  }
  \end{tabular}
  &
  \begin{tabular}{@{}c@{}}
  \scalebox{0.9}{
  \begin{tikzpicture}[>=latex,scale=0.8, every node/.style={scale=0.7}]
    \tikzstyle{state}=[circle,draw,trans, minimum size=8mm]
    \tikzstyle{initstate}=[circle,draw,trans, minimum size=8mm, fill=lightgrey]
    \tikzstyle{trans}=[font=\footnotesize]

    \path (0,0) node[initstate] (q00) {$(0,0)$}
          (0,2) node[state] (q11) {$(1,1)$}
		  						+(0,1.1) node {$\prop{pol_1}$}
		  						+(0,0.7) node {$\prop{pol_2}$}
          (2,0) node[state] (q22) {$(2,2)$}	
		  (2,2) node[state] (q33) {$(3,3)_1$}
		  						+(0.9,0) node {$\prop{pol_1}$}
		  						+(0.9,-0.4) node {$\prop{pol_2}$}
		  						+(0.9,-0.8) node {$\prop{target}$}
		  (-2,-2) node[state] (q33v) {$(3,3)_2$}
		  						+(0,-0.7) node {$\prop{pol_1}$}
		  						+(0,-1.1) node {$\prop{pol_2}$} 			  												
		  						+(0,-1.5) node {$\prop{allvisited}$}					
		  						+(0,-1.9) node {$\prop{target}$}					
		  (0,-2) node[state] (q12) {$(1,2)$}
		  						+(0.9,0.1) node {$\prop{pol_1}$}
		  (-2,0) node[state] (q21) {$(2,1)$}
		  						+(0,0.7) node {$\prop{pol_2}$}					  	
        ;

   \path[->,font=\scriptsize] (q00)
         edge
           node[midway,sloped]{\onlabel{(N,N)}} (q11)
         edge
           node[midway,sloped]{\onlabel{(E,E)}} (q22)
         edge
           node[midway,sloped]{\onlabel{(N,E)}} (q12)
         edge
           node[midway,sloped]{\onlabel{(E,N)}} (q21);
   \path[->,font=\scriptsize] (q11)
         edge
           node[midway,sloped]{\onlabel{(E,E)}} (q33);
   \path[->,font=\scriptsize] (q22)
         edge
           node[midway,sloped]{\onlabel{(N,N)}} (q33);
   \path[->,font=\scriptsize] (q12)
         edge
           node[midway,sloped]{\onlabel{(E,N)}} (q33v);
   \path[->,font=\scriptsize] (q21)
         edge
           node[midway,sloped]{\onlabel{(N,E)}} (q33v);

   \draw[-latex,black](q33) ..controls +(1.2,0.6) and +(0,1.4).. (q33);
   \draw[-latex,black](q33v) ..controls +(-1.2,0.6) and +(-1.2,-0.6).. (q33v);

  \end{tikzpicture}
  }
  \end{tabular}
\\
$(\Mmulti)_{f_{\top}}$\quad and\quad $(\Mmulti)_{f_{\top\!_g}}$ &  $(\Mmulti)_{f_{\top\!_d}}$
\end{tabular}

\begin{tabular}{c}
  \begin{tabular}{@{}c@{}}
  \scalebox{0.9}{
  \begin{tikzpicture}[>=latex,scale=0.8, every node/.style={scale=0.7}]
    \tikzstyle{state}=[circle,draw,trans, minimum size=8mm]
    \tikzstyle{initstate}=[circle,draw,trans, minimum size=8mm, fill=lightgrey]
    \tikzstyle{trans}=[font=\footnotesize]

    \path (0,0) node[initstate] (q00) {$(0,0)$}
		  						+(0.7,0.5) node {$\prop{pol_2}$}
		  						+(0.7,0.85) node {$\prop{pol_1}$}
          (0,2) node[state] (q11) {$(1,1)$}
		  						+(0,1.1) node {$\prop{pol_1}$}
		  						+(0,0.7) node {$\prop{pol_2}$}
          (2,0) node[state] (q22) {$(2,2)$}	
		  (2,2) node[state] (q33) {$(3,3)_1$}
		  						+(0.9,0) node {$\prop{pol_1}$}
		  						+(0.9,-0.4) node {$\prop{pol_2}$}
		  						+(0.9,-0.8) node {$\prop{target}$}
		  (-2,-2) node[state] (q33v) {$(3,3)_2$}
		  						+(0,-0.7) node {$\prop{pol_1}$}
		  						+(0,-1.1) node {$\prop{pol_2}$} 			  												
		  						+(0,-1.5) node {$\prop{allvisited}$}					
		  						+(0,-1.9) node {$\prop{target}$}					
		  (0,-2) node[state] (q12) {$(1,2)$}
		  						+(0.9,0.1) node {$\prop{pol_1}$}
		  (-2,0) node[state] (q21) {$(2,1)$}
		  						+(0,0.7) node {$\prop{pol_2}$}					  	
        ;

   \path[->,font=\scriptsize] (q00)
         edge
           node[midway,sloped]{\onlabel{(N,N)}} (q11)
         edge
           node[midway,sloped]{\onlabel{(E,E)}} (q22)
         edge
           node[midway,sloped]{\onlabel{(N,E)}} (q12)
         edge
           node[midway,sloped]{\onlabel{(E,N)}} (q21);
   \path[->,font=\scriptsize] (q11)
         edge
           node[midway,sloped]{\onlabel{(E,E)}} (q33);
   \path[->,font=\scriptsize] (q22)
         edge
           node[midway,sloped]{\onlabel{(N,N)}} (q33);
   \path[->,font=\scriptsize] (q12)
         edge
           node[midway,sloped]{\onlabel{(E,N)}} (q33v);
   \path[->,font=\scriptsize] (q21)
         edge
           node[midway,sloped]{\onlabel{(N,E)}} (q33v);

   \draw[-latex,black](q33) ..controls +(1.2,0.6) and +(0,1.4).. (q33);
   \draw[-latex,black](q33v) ..controls +(-1.2,0.6) and +(-1.2,-0.6).. (q33v);

  \end{tikzpicture}
  }
  \end{tabular}
\\
$(\Mmulti)_{f_{\bot_d}}$,\quad $(\Mmulti)_{f_{\bot_g}}$,\quad and\quad $(\Mmulti)_{f_{\bot_d\meet\bot_g}}$
\end{tabular}

\caption{Model translations for the drone scenario}
\label{fig:drones-translation2}

\vspace*{9mm}

\small
\centering
\begin{myalgorithm}{$gmcheck_{tr}(M,\varphi)$}\vspace{-2mm}
\itemsep=0.4pt
\item Set the initial valuation of $\varphi$ to $\V_\varphi$ such that $\V_\varphi(q)=\bot$ for every $q\in\States$;
\item For every join-irreducible logical value $\ell\in\JI(\L)$:
  \begin{itemize}
  \item[$\bullet$] Compute the set of states $Q_\ell$ that (classically) satisfy $\varphi$ in $M_{f_\ell}$;
  \item[$\bullet$] For each $q\in Q_\ell$, do $\V_\varphi(q) := \V_\varphi(q) \join \ell$;
  \end{itemize}
\item Return $\V_\varphi$.
\end{myalgorithm}\vspace{-1mm}
\caption{Translation-based global model checking for \mvATL}
\label{fig:global-mcheck}\vspace*{-4mm}
\end{figure}

The algorithm in Figure~\ref{fig:mcheck-translation} is an example of \emph{local} model checking.
That is, given a state (respectively, path) and a formula, it returns the truth value of the formula in that state (respectively, on that path).
In two-valued modal logics, verification of state formulas is often done by means of \emph{global model checking}
that returns the exact set of states where the input formula holds.
For many logics -- including \ATL and \ATLs -- this provides strictly more information with no extra computational cost.
The analogous problem for multi-valued modal logics would ask for a \emph{valuation} of the input formula,
i.e., a mapping from the states of the model to the truth values of $\varphi$.
A global model checking algorithm for \mvATL, based on the translation to two-valued model checking, is presented in the next section, in Figure~\ref{fig:global-mcheck}.

\subsection{Translating implication formulas: Impossibility result}\label{sec:impossibility}

Unfortunately, Theorem \ref{general} cannot be extended to \ATLSi, i.e., to the full language containing implication formulas of the form
$\varphi_1 \ra \varphi_2$, where $\ra$ represents the lattice order.

\begin{proposition}
There are lattice reduction triples and formulas of \ATLSi that
do not satisfy translation condition (\ref{thesis}).
\end{proposition}

\begin{proof}
Consider the lattice $\L_o=(\LL_o, \leq)$, where $\LL_o = \{0,..., k-1, k,..., 2k-1\}$, and $\leq$ is the usual total order on $\LL_o$.
Clearly, in $\L_o$ we have $0=\bot$ and $2k-1=\top$ according to our lattice notation.
Then $(\{0,2k-1\}, \leq)$ is a sublattice of $\L_o$ and  the reduction $f: \LL_o \ra \{0,2k-1\}$ given by $f(x)=2k-1$  if $x\geq k$, and $f(x)= 0$ if  $x<k$
preserves the bounds in $\L_o$. Thus $(\L_o, \{0,2k-1\}, f)$ is a lattice reduction triple.

Now take arbitrary $k_1, k_2$ such that $0<k_1<k_2<k$, and an mv-CGS $M$
over $\L_o^+= (\L_o, \sigma)$ for an arbitrary  $\sigma:\LConst \rightarrow \LL_o$
such that, for some state $q \in \States$ of $M$ and atomic propositions $p_1,p_2\in \AP$, we have $V(p_i,q) = k_i$ for $i=1,2$.

Next, let $\varphi = p_2 \ra p_1$.
Since $\truthvalue{p_i}{M,q} = k_i$ for $i=1,2$ and $k_2 > k_1$, we have
$\neg (\truthvalue{p_2}{M,q} \leq  \truthvalue{p_1}{M,q})$,  whence  $\truthvalue{\varphi}{M,q} = 0$.

However, for the model $M_1$ obtained from $M$ with the reduction $f$
we get $ \truthvalue{p_i}{M_1,q} = 0$ for $i=1,2$ (as $k_i< k$, $f(k_i) = 0$
for $i=1,2$), where $\truthvalue{p_2}{M_1,q}\leq \truthvalue{p_1}{M_1,q}$, which implies
$\truthvalue{\varphi}{M_1,q} = 2k-1$.
Yet, as  $f^{-1}(2k-1) = \{k, k+1,..., 2k-1\}$
we have  $\truthvalue{\varphi}{M,q} = 0 \not\in f^{-1}(2k-1)$, which contradicts Equation (\ref{thesis}).
\end{proof}

The above result can be generalized as follows:

\begin{theorem}
\label{counterexample}
If $\L = (\LL, \leq)$ of Theorem~\ref{general} contains a chain or anti-chain of cardinality $n$,
and $\L'=(\LL', \leq')$ is a sublattice of $\L$ of cardinality $n'< n$, then there is no function
$f:\LL \ra \LL'$ satisfying translation condition (\ref{thesis}) if the language under consideration contains implication formulas.
\end{theorem}
\begin{proof}
Assume $X= \{x_1, x_2, \ldots, x_n\}$ is a chain or anti-chain in $L$.
Let $M$ be a model such that, for some state $s\in Q$ of $M$ and propositional
variables $p_1,p_2, \ldots, p_n\in \AP$, we have $V(s, p_i)=k_i$ for $i=1,2, \ldots, n$.
Consider any $f:L\ra L'$. As $card(L')= n'<n$, than there must be $k,l \in \{1,\ldots, n\}$
such that $k\neq l$ and $f(x_k)= f(x_l)$. Let
$$\varphi = \left\{
\begin{array}{ll}p_l\ra p_k & \mbox {if $X$ is a chain and $x_k< x_l$}\\
p_k\ra p_l & \mbox {otherwise}
\end{array}
\right.$$

Then, clearly, $ x_l\not\leq x_k$ (note that if $X$ is an anti-chain,
then $ x_r\not\leq x_q$  for any $1 \leq r,s\leq n$).
Hence $\truthvalue{p_l}{M,q}\not\leq  \truthvalue{p_k}{M,q}$, and  $\truthvalue{\varphi}{M,q} = \bot$.
However, as $f(x_k)= f(x_l)$, we have $\truthvalue{p_l}{M_1,q}\leq  \truthvalue{p_k}{M_1,q}$,
whence $\truthvalue{\varphi}{M_1,q} = \top$. Thus $f$ does not satisfy Equation (\ref{thesis}).
\end{proof}

In other words, if the size of the target lattice $\L'$ is strictly smaller than the ``diameter'' of the source lattice $\L$ (that is, the cardinality of the longest chain or antichain in $\L$), then a ``clustering'' of truth values from $\L$ into $\L'$ that preserves Equation (\ref{thesis}) is impossible.
Note that the diameter of \emph{any} lattice with more than 2 values must be at least 3, and hence it exceeds the size of the classical 2-valued lattice \textbf{2}.
The following is an immediate consequence of the above:

\begin{corollary}
\label{counterexample-2val}
For any multi-valued lattice $\L$
there is no reduction of $\L$ to the 2-valued lattice of classical truth values that satisfies
the translation condition (\ref{thesis}) for the whole language of \ATLSi.
\end{corollary}

\subsection{Translating implication formulas even more impossible}\label{sec:more-impossibility}

We already know that there is no general translation from multi-valued to two-valued model checking for implication formulas.
Now we will show that the impossibility result can be extended to any ``clustering'' of truth values from the original lattice $\L$.
To this end, we give a necessary and sufficient condition for the existence of a function $f: L \ra L_f$
that preserves bounds in $\L$ and satisfies the translation condition (\ref{thesis})
also for {implication} formulas.

The following lemma states that, in order to obtain the analogue of Theorem~\ref{general}
for implication formulas, the mapping $f$ would have to preserve both the ordering
and incomparability of the elements in $\L$.

\begin{lemma}
\label{fororder}
Let $(\L, \L_f,f)$ be a lattice reduction triple,
let $M$ be an \mvmodels over $\LL$, and $f(M)$ its reduction to $\L_f$.
Then translation condition (\ref{thesis}) of that theorem is satisfied for all implication formulas iff  the following conditions hold:
\begin{enumerate}
 \item[C1:] $(\forall x_1, x_2 \in \LL)\ [x_1 < x_2\ \Rightarrow\  f(x_1)< f(x_2)]$
 \item[C2:] $(\forall x_1, x_2 \in \LL)\ [x_1 \bowtie  x_2\  \Rightarrow\   f(x_1)\bowtie  f(x_2)]$
\end{enumerate}
\end{lemma}

\begin{proof}
Note that an implication  formula is a state formula, and that for any such formula
$\psi$ we have $\truthvalue{\psi}{M,q} \in \{\bot, \top\}$ for any mv-CGS $M$ and
any $q \in \States$.
Thus in order to prove (\ref{thesis}) for such formulas,
it suffices to show that, for any implication formula $\varphi$ and any state $q$:

\begin{equation}
\label{thesis3}
\truthvalue{\varphi}{M_f,q} = \top \mbox{ iff } \truthvalue{\varphi}{M,q} \in f^{-1}(\top)
\end{equation}

\smallskip\noindent
\textbf{\underline{``(\ref{thesis3}) $\pmb{\Rightarrow}$ (C1 $\land$ C2)''}:}\quad
We start by proving the necessity of conditions C1 and C2.
Assume first that $f$ satisfies  (\ref{thesis3}) for implication formulas.
We should prove that $f$ satisfies Conditions C1, C2  for all such formulas.
For what follows, denote $\xi= p_1 \ra p_2, \psi = p_2 \ra p_1$,
where $p_1,p_2\in \AP$, $p_1\neq p_2$  where $\AP$ is the set of atomic propositions of our logic).

\begin{description}
\item[C1:] We argue by contradiction.
Suppose  $x_1, x_2\in \LL, x_1<x_2$ and $f(x_1)\geq f(x_2)$.
Then, as $x_1 < x_2$ implies $x_1\leq x_2$ and $f$ preserves bounds,
we also have $f(x_1)\leq  f(x_2)$, whence $f(x_1)= f(x_2)$.

Now let $M$ be an mv-CGS over $\L$ such that, for some state $q\in \States$,
we have $V(p_i,q) = x_i$ for $i=1,2$, and let $M_f$ be the image of $M$ under $f$.
Since $\psi = p_2 \ra p_1$ and $\truthvalue{p_2}{M,q} = x_2 > x_1 = \truthvalue{p_1}{M,q}$,
we have $\truthvalue{\psi}{M,q} = \bot$.
However,  $\truthvalue{\psi}{M_f,q} =\top$, because $\truthvalue{p_2}{M_f,q} = f(x_2) = f(x_1) = \truthvalue{p_1}{M_f,q}$.
As $\bot\not\in f^{-1}(\top)$, this contradicts (\ref{thesis3}).

\item[C2:] We again argue by contradiction.
Suppose $x_1, x_2\in L, x_1\bowtie x_2$ and $\neg (f(x_1)\bowtie f(x_2))$.
Without any loss of generality, we can assume that $f(x_1)\leq f(x_2)$.
Let mv-CGSs $M,M_f$ and state $q$ of $M$ be like in the preceding item.
Then, as  $\truthvalue{p_1}{M,q} = x_1 \bowtie  x_2 =\truthvalue{p_2}{M,q}$, we have in particular
$\truthvalue{p_1}{M,q}\not \leq  \truthvalue{p_2}{M,q}$.
Since $\xi = p_1\ra p_2$, this implies  $\truthvalue{\xi}{M,q} = \bot$.
In turn, $\truthvalue{\xi}{M_f,q} = \top$, because $\truthvalue{p_1}{M_f,q} = f(x_1)\leq f(x_2) = \truthvalue{p_2}{M_f,q}$,
which again contradicts (\ref{thesis3}).
\end{description}


\smallskip\noindent
\textbf{\underline{``(C1 $\land$ C2) $\pmb{\Rightarrow}$ (\ref{thesis3})''}:}\quad
It remains to prove the sufficiency of conditions C1 and C2.
We assume that C1, C2 hold, and prove that (\ref{thesis3})
holds for formulas of the form $\varphi = \varphi_1 \ra \varphi_2$.
We start by proving this result for non-nested implication formulas,
i.e., we assume that $\varphi_1, \varphi_2$ do not contain $\ra$.
Then, by Theorem \ref{general}, (\ref{thesis3}) holds for
$\varphi_1, \varphi_2$, which implies that

\begin{equation}
\label{thesis4}
\truthvalue{\varphi_i}{M_f,q} = f(\truthvalue{\varphi_i}{M,q}), \ \ i=1,2.\ \ \ \
\end{equation}

\smallskip\noindent
\underline{``(\ref{thesis3}L) $\Rightarrow$ (\ref{thesis3}R)''}:\quad
We begin with the forward implication in (\ref{thesis3}).
Assume that $\truthvalue{\varphi}{M_f,q} = \top$.
Then $\truthvalue{\varphi_1}{M_f,q}\leq \truthvalue{\varphi_2}{M_f,q}$.
By (\ref{thesis4}), this implies $f(\truthvalue{\varphi_1}{M,q})\leq f(\truthvalue{\varphi_2}{M,q}$.
We show by contradiction that

\begin{equation}
\label{thesis1*}
\truthvalue{\varphi_1}{M,q}\leq \truthvalue{\varphi_2}{M,q}.
\end{equation}
Suppose that (\ref{thesis1*}) does not hold, then we have two possible cases:
\begin{description}
\item[Case 1:]
$\truthvalue{\varphi_1}{M,q}> \truthvalue{\varphi_2}{M,q}$.
Then by C1 we have
$f(\truthvalue{\varphi_1}{M,q}) >  f(\truthvalue{\varphi_2}{M,q})$,
whence from (\ref{thesis4}) we get $\truthvalue{\varphi_1}{M_f,q} > \truthvalue{\varphi_2}{M_1,q}$
and $\truthvalue{\varphi}{M_f,q} = \bot$ --- which is a contradiction.

\item[Case 2:]
$\truthvalue{\varphi_1}{M,q}\bowtie  \truthvalue{\varphi_2}{M,q}$.
Then $f(\truthvalue{\varphi_1}{M,q}) \bowtie  f(\truthvalue{\varphi_2}{M,q})$ by C2,
whence from (\ref{thesis4})
we get $\neg (\truthvalue{\varphi_1}{M_f,q} \leq \truthvalue{\varphi_2}{M_f,q})$.
Consequently,
$\truthvalue{\varphi}{M_f,q} = \bot$ --- which is again a contradiction.
\end{description}
Thus (\ref{thesis1*}) above holds, whence $\truthvalue{\varphi}{M,q} = \top \in f_1(\top)$,
and  the forward implication in (\ref{thesis3}) holds.

\smallskip\noindent
\underline{``(\ref{thesis3}R) $\Rightarrow$ (\ref{thesis3}L)''}:\quad
The final step is proving the backward implication in (\ref{thesis3}).
Assume that $\truthvalue{\varphi}{M,q} = f^{-1}(\top)$.
As $\truthvalue{\varphi}{M,q}\in \{\bot, \top\}$ and $f(\bot)\neq \top$ by the preservation
of bounds by $f$ and the non-triviality of $L, L_f$, we obtain $\truthvalue{\varphi}{M,q}=\top$,
whence $\truthvalue{\varphi_1}{M,q}\leq \truthvalue{\varphi_2}{M,q}$.
Since $f$ preserves bounds, this implies $f(\truthvalue{\varphi_1}{M,q}) \geq  f(\truthvalue{\varphi_2}{M,q})$,
whence from (\ref{thesis4}) we obtain $\truthvalue{\varphi_1}{M_1,q} \leq \truthvalue{\varphi_2}{M_f,q}$.
This yields $\truthvalue{\varphi}{M_f,q} = \top$, whence the  backward implication in (\ref{thesis3})
holds, too.

\smallskip\noindent
\underline{Nested formulas}:\quad
The proof for nested formulas proceeds by induction.
Assume that (\ref{thesis3}) holds for implication formulas with $\ra$ nested
at most $k$ times, and assume $\varphi$ is an implication formula with $\ra$ nested $k+1$ times.
Then  $\varphi = \varphi_1 \ra \varphi_2$, where $\ra$ is nested at most
$k$ times in $\varphi_1, \varphi_2$.
Consequently, by the inductive assumption (\ref{thesis1*}) holds for $\varphi_1, \varphi_2$,
and repeating the proof given above for implication formulas without nesting
of $\ra$ we can show that  (\ref{thesis3}) holds for $\varphi$ too.

\medskip
This completes the proof of the sufficiency of C1, C2 for all implication formulas,
and the proof of Lemma \ref{fororder}.
\end{proof}

From Lemma \ref{fororder} we can easily derive by induction the following general result:

\begin{theorem}
\label{th-negative}
Let the $(\L,\L_f, f)$ be an LRT.
Then, the translation condition
is satisfied for all formulas of \ATLSi over $\L$ iff conditions C1, C2 of Lemma \ref{fororder} hold.
\end{theorem}

It can be seen that conditions C1, C2 imply that any translation $f$ meeting them must preserve the exact structure of the lattice $\L$.
An important consequence of that fact is:

\begin{corollary}\label{prop:onetoone}
Given a lattice $\L = (\LL, \leq)$ and its sublattice
$\L_f= (\LL_f, \leq_f)$, any function $f:\LL \ra \LL_f$ preserving
the \lattice\ bounds and satisfying translation condition
(\ref{thesis}) for all implication formulas must be one-to-one.
\end{corollary}

\begin{proof}
Suppose that $f$ satisfies the above assumption, $x_1, x_2\in L$ and $x_1\neq x_2$.
Then we have one of the following cases:
\begin{enumerate}
\item
$x_1 < x_2$ or $x_2 < x_1$.
Then  $f(x_1) \neq f(x_2)$ by C1 of Theorem \ref{th-negative}.
\item
$x_1\bowtie x_2$.
Then  $f(x_1)\bowtie  f(x_2)$ by C2 of Theorem \ref{th-negative},
which again implies $f(x_1) \neq f(x_2)$.
\end{enumerate}

\vspace*{-5mm}
\end{proof}

The meaning of Corollary~\ref{prop:onetoone} is that there is no way of reducing
$n$-valued model checking to $k$-valued model checking for $k<n$, if we want to handle
all implication formulas.
Clearly, Corollary~\ref{counterexample-2val} in Section~\ref{sec:impossibility} is a special case of the above result.

\subsection{Translation of model checking for \emph{some} implication formulas}

By Corollary~\ref{prop:twovalued}, there is a simple translation from multi-valued to classical model checking
for strategic and temporal operators.
By Corollary~\ref{prop:onetoone}, we know that it cannot be generally extended to implication formulas.
The next question is: can we construct such a translation for \emph{some} implication formulas?
If so, for which ones?
The impossibility result in Corollary~\ref{prop:onetoone} is due to the fact that implication
formulas can be used to encode the semantics in the language --- including in particular
its $n$-valued character.
However, one usually wants to model-check one formula at a time.
Then, Theorem \ref{th-negative} can be in some cases modified to provide the desired reduction:

\begin{theorem}
\label{th-positive}
Let $(\L, \LL_f,f)$ be a lattice reduction triple (LRT), let $M$ be a CGS over
$\L$ and $f(M)$ its reduction to $\L_f$.
Further, let $\varphi$ be a formula of \ATLSi and $Sub(\varphi)$ the set of all its subformulas.
Then $\varphi$ satisfies translation condition (\ref{thesis})
whenever, for any implication formula $\phi\in Sub(\varphi)$ such that $\phi=\varphi_1\ra\varphi_2$,
any state (resp.~path) $\xi$,
and $x_i = \truthvalue{\varphi_i}{M,\xi}, i=1,2,$ the following conditions hold:

\begin{center}
\begin{tabular}{llll}
 {\bf C1':} $x_1 < x_2\ \Rightarrow\  f(x_1)< f(x_2)$  &&&
{\bf C2':}  $x_1 \bowtie   x_2\  \Rightarrow\   f(x_1) >  f(x_2)$
\end{tabular}
\end{center}
\end{theorem}

\begin{proof}
To prove the thesis, we assume that C1', C2' are satisfied, and show by
structural induction that  translation condition (\ref{thesis})
\[
\truthvalue{\psi}{M_f,\xi} = x \ \ \ \  {\rm iff}\ \ \ \ \  \truthvalue{\psi}{M,\xi}\in f^{-1}(x)
\]
holds for any $\psi \in Sub(\varphi)$.

 \medskip
For atomic or constant $\psi$, the thesis follows from Theorem \ref{general}.
Suppose now that Equation~(\ref{thesis}) holds for all subformulas of $\varphi$ having rank $k$,
and assume $\psi$ is of rank $k+1$.
If $\psi$ is obtained from subformulas of rank at most $k$ using any operator $Op$
other than $\ra$, then Equation~(\ref{thesis}) follows from the fundamental
Lemma \ref{prop:lem-gen}.

\medskip
Thus, it remains to consider the case of $\ra$.
Assume $\psi= \psi_1\ra \psi_2$, with Equation~(\ref{thesis}) being satisfied for both $\psi_1, \psi_2$.
Since $\psi$ is an implication formula, according to what we have already
noted in the proof of Lemma~\ref{fororder}, proving (\ref{thesis}) for $\psi$ reduces to showing condition (\ref{thesis3}), i.e.,
\[
\truthvalue{\psi}{M_f,q} = \top \ \ \ \  {\rm iff}\ \ \ \ \  \truthvalue{\psi}{M,q}\in f^{-1}(\top).
\]
Note that since $\psi_1, \psi_2$ are  in  $Sub(\varphi)$,
then C1', C2' hold for $x_i= \truthvalue{\varphi_i}{M,q}, i=1,2$.
By the inductive assumption, we also have
\begin{equation}
\label{thesis7}
\truthvalue{\psi_i}{M_f,q} = f(\truthvalue{\psi_i}{M,q}), \ \ i=1,2\ \vspace*{1mm}
\end{equation}
\Ra
We begin with the forward implication in
(\ref{thesis3}).
Assume that $\truthvalue{\psi}{M_f,q} = \top$.
Then $\truthvalue{\psi_1}{M_f,q}\leq \truthvalue{\psi_2}{M_f,q}$.
By (\ref{thesis7}),
this implies $f(\truthvalue{\psi_1}{M,q})\leq f(\truthvalue{\psi_2}{M,q}$.
We show by contradiction that it implies
\begin{equation}
\label{thesis2*}
\truthvalue{\psi_1}{M,q}\leq \truthvalue{\psi_2}{M,q}
\end{equation}
Suppose that (\ref{thesis2*}) does not hold, then we have two possible cases: \vspace*{1.6mm}

Case 1: $\truthvalue{\psi_1}{M,q}> \truthvalue{\psi_2}{M,q}$.
Then by condition C1' we have  $f(\truthvalue{\psi_1}{M,q}) >  f(\truthvalue{\psi_2}{M,q})$,
whence from (\ref{thesis7})
we get $\truthvalue{\psi_1}{M_f,q} > \truthvalue{\psi_2}{M_f,q}$ and $\truthvalue{\psi}{M_f,q} = \bot$,
which is a contradiction.

Case 2: $\truthvalue{\psi_1}{M,q}\bowtie  \truthvalue{\psi_2}{M,q}$.
Then $f(\truthvalue{\psi_1}{M,q}) >   f(\truthvalue{\psi_2}{M,q})$ by condition C2',
which again leads to a contradiction by what we have already proved for Case 1.

Thus (\ref{thesis2*}) above holds, whence $\truthvalue{\psi}{M,q} = \top \in f^{-1}(\top)$,
and  the forward implication in (\ref{thesis3})
holds. \vspace*{1mm}

\La
The final step consists in proving the backward implication in (\ref{thesis3}).
Assume that $\truthvalue{\psi}{M,q} = f^{-1}(\top)$.
As $\truthvalue{\psi}{M,q}\in \{\bot, \top\}$ and $f(\bot)\neq \top$ by the preservation
of bounds by $f$ and the non-triviality of $\L,\L_f$, we get $\truthvalue{\psi}{M,q}\!=\!\top$,
and consequently  $\truthvalue{\psi_1}{M,q}\leq \truthvalue{\psi_2}{M,q}$.
Since $f$ preserves bounds, this implies $f(\truthvalue{\psi_1}{M,q}) \leq  f(\truthvalue{\psi_2}{M,q})$,
whence from (\ref{thesis7}) we obtain $\truthvalue{\psi_1}{M_f,q} \leq \truthvalue{\psi_2}{M_f,q}$.
This yields  $\truthvalue{\psi}{M_f,q} = \top$, whence the backward implication
in (\ref{thesis3}) holds, too.
\end{proof}

Assume that our \mvmodels\ are defined over distributive lattices.
We now show that the translation method of Section~\ref{sec:mcheck-mv-translation},
based on join irreducible elements $\JI(\L)$, can be applied to a formula $\varphi$
of \ATLs and an \mvmodel\ $M$, provided that the assumptions of Theorem \ref{th-positive} are satisfied.
By (\ref{decomp}), for each $x \in \LL$ we have $x = \bigjoin (\JI(\L)\; \cap \downarrow x)$.
Let $M^\ell$ be the model obtained using the translation $f_\ell$.
Therefore, according to Theorem \ref{th-positive}:
$
\truthvalue{\varphi}{M^\ell,\xi} = x \mbox{ iff } \truthvalue{\varphi}{M,\xi} \in f_\ell^{-1}(x)
$
whence $\truthvalue{\varphi}{M^\ell,\xi} = \top \mbox{ iff } \truthvalue{\varphi}{M,\xi} \in\; \uparrow\! \ell$.
Thus,
\begin{equation}
\label{for2}
\truthvalue{\varphi}{M,\xi} = \bigjoin\{ \ell \in \JI(\L) \;\mid \; \truthvalue{\varphi}{M^\ell,\xi} = \top\}.
\end{equation}

\begin{example}
Consider model $\Mmulti$ in Figure~\ref{fig:drones-mv} and formula $\phi = \coop{1}\Always(\prop{pol_1}\ra(\prop{target}\land\prop{pol_2}))$.
Subformula $\prop{pol_1}$ can take the following truth values throughout the model: $\bot, u, \top\!_d, \top$.
Similarly, $\prop{target}\land\prop{pol_2}$ can evaluate to $\bot, \top\!_d$.
Thus by Theorem~\ref{th-positive} mapping $f_{\top\!_d}$ meets translation condition~(\ref{thesis}),
and we can use the translation method of Section~\ref{sec:mcheck-mv-translation} to check if the value of $\phi$ is at least~$\top\!_d$.

On the other hand, all the other ``cutoff'' mappings (i.e.,
$f_{\bot_d\meet\bot_g}, f_{\bot_d}, f_{\bot_g}, f_{\top\!_g}$, and $f_\top$) do not satisfy condition C1',
and hence the correctness of the translation is not guaranteed for those truth values.
\end{example}

The following is an immediate consequence of Theorem~\ref{th-positive}.

\begin{corollary}
Let $\L$, its sublattice $\L_f$, $f:\LL \ra \LL_f$, and $M, M_f$ be as in Theorem \ref{general}.
Further, let $\varphi$ be a formula of \ATLSi such that every implication subformula of $\varphi$
is of the form $\psi_1\ra\psi_2$, where $\psi_i\in\{ \const{\bot}, \const{\top}\}$ for some $i\in\{1,2\}$.
Then $\varphi$ satisfies the translation condition (\ref{thesis}).
\end{corollary}

\begin{proof}
It suffices to observe that if at least one of the formulas $\psi_1, \psi_2$ is either $\const{\bot}$ or $\const{\top}$, then the implication subformula $\psi_1\ra\psi_2$ trivially satisfies conditions C1' and C2' of Theorem~\ref{th-positive}.
\end{proof}

\subsection{Recursive model checking of \ATLs}\label{sec:recursive}

\begin{figure}[!b]\small
\centering
\begin{myalgorithm}{$gmcheck_{rec}(M,\varphi)$}\vspace{-2mm}
\itemsep=0.4pt
\item If $\varphi$ contains no instance of the comparison operator $\ra$, then return $gmcheck_{tr}(M,\varphi)$;
\item Else, pick the first implication subformula $\varphi_1\ra\varphi_2$ in $\varphi$, and:
  \begin{itemize}
  \item[$\bullet$] Compute $\V_{\varphi_1} := gmcheck_{rec}(M,\varphi_1)$ and $\V_{\varphi_2} := gmcheck_{rec}(M,\varphi_2)$;
  \item[$\bullet$] Create an extension $M'$ of model $M$ by adding a fresh atomic proposition $\prop{p}$,
	                 and fixing its valuation so that $V(\prop{p},q) = \top$ if $\V_{\varphi_1}(q) \le \V_{\varphi_2}(q)$ and $V(\prop{p},q) = \bot$ otherwise;
  \item[$\bullet$] Create formula $\varphi'$ by replacing every occurrence of $\varphi_1\ra\varphi_2$ by $\prop{p}$;
  \item[$\bullet$] Return $gmcheck_{rec}(M',\varphi')$.
  \end{itemize} \vspace{-2mm}
\end{myalgorithm}
\caption{Recursive global model checking for \ATLs}
\label{fig:recursive-mcheck}
\end{figure}

For many model checking instances, the assumptions of Theorem \ref{th-positive} do not hold.
In those cases, we cannot translate the multi-valued model checking of \ATLSi formulas
to the classical model checking for \ATLs.
Then, the simplest solution is to adapt the standard recursive algorithm that,
in order to model-check formula $\varphi$, proceeds bottom-up from the simplest subformulas, and replaces them with fresh atomic propositions.
In our case, this means that model checking of each implication formula $\varphi_1 \ra \varphi_2$
consists in computing the values of $\varphi_1, \varphi_2$ by means of the translation
in Section~\ref{sec:mcheck-mv-translation}, and then fixing the valuation
of the fresh variable $\prop{p}_{\varphi_1 \ra \varphi_2}$ according to their comparison.
The detailed algorithm is presented in Figure~\ref{fig:recursive-mcheck}.

The main disadvantage of the above method compared to the direct translation method is that it requires computing the values of the new atomic
propositions for all states of the model $M$.
In other words, we need to carry out global model checking, whereas for formulas without the implication operator $\ra$ both global and local model checking were possible.
The method can be possibly improved if we assume that a specific symbolic model checking method for two-valued \ATLs\ is used; we leave a study of this subject for future work.

Nevertheless, the algorithm presented in Figure ~\ref{fig:recursive-mcheck} has two important consequences.
First, it provides a general linear-time reduction from model checking \ATLSi\ (resp.~\ATLi)
to model checking standard 2-valued \ATLs (resp.~\ATL).
We state it formally as follows.

\begin{theorem}
The one-to-many reduction from multi-valued model checking of \ATLSi\ to 2-valued model checking of \ATLs runs
in linear time with respect to the size of the model, the length of the formula, and the number of truth values.
\end{theorem}

\begin{corollary}
Model checking \ATLSi\ (resp.~\ATLi) is \DExptime-complete (resp.~\Ptime-complete) in the size of the model, the length of the formula, and the number of truth values.
\end{corollary}

Secondly, we note that correctness of the translation does not depend on the type of strategies being used in the semantics of \ATLSi.
As it is, the translation provides a model checking reduction to the \IR variant of \ATLs (perfect information + perfect recall).
If we used memoryless strategies of type $s_a:\States\to\Actions$ instead of perfect recall, the translation would yield reduction to
the \Ir variant of \ATLs (perfect information + imperfect recall~\cite{Schobbens04ATL}).
Since the \IR and \Ir semantics coincide in 2-valued \ATL (though not in \ATLs!), we get the following.

\begin{theorem}
For \ATLonostar, memory is irrelevant, i.e., its semantics can be equivalently given by memoryless strategies.
\end{theorem}


\section{Multi-valued transitions}\label{sec:mvtrans}

In this paper our aim is to propose a framework for a graded interpretation of logical statements referring to the strategic ability of agents and coalitions.
Until this point, the ``graded'' truth values have only originated from non-classical interpretation of atomic propositions and literals.
Typically, this happens because when constructing a model  we cannot determine the truth of some basic statements in absolute terms
(as either true or false).
Instead, we assign such basic statements with their ''truth degrees'' (which can be also seen as ``weights of evidence'')
drawn from a suitable lattice, which then propagate to more complex formulas.

Another source of non-classical truth values sometimes considered in the literature is a graded interpretation of transitions.
In that case, each transition is labelled according to its ``strength.''
An extension of \ATLSi with weighted transitions is discussed in this section.

\subsection{Weighted transitions: Potential interpretations}

The shift from 2-valued to multi-valued modal logic typically arises when we extend the domain of interpretation for atomic propositions
in \emph{states} of the model. The level of truth for $\prop{p_1}, \prop{p_2}, \dots$, is not crisp anymore, and this propagates
to more complex formulae $\varphi$ via semantic clauses.
So far, we have assumed that the {transition relation} {is} crisp, i.e., given states $q,q'$ and a vector of actions $\vec{\alpha}$,
the transition from $q$ to $q'$ labeled by $\vec{\alpha}$ is either fully included in the model, or is completely absent from it.
An alternative would be to consider multi-valued transition relations, with transitions that are possible to a certain degree.

\medskip
There are at least two sensible interpretations of such weighted transitions. On the one hand, the weight can be interpreted
as the strength of evidence supporting the existence of the transition.
This approach has been adopted in the previous works on multi-valued temporal logics over arbitrary lattices of truth values~\cite{kp02a,kp06},
with the additional assumption that the weights of transitions are drawn from the same lattice as the values of propositions.
A characteristic feature of the semantics in~\cite{kp02a,kp06} is that, whenever the weights on transitions decrease sufficiently,
the value of a temporal formula must also decrease.
Formally, consider a multi-valued transition system $M$, a state $q$ in $M$, and a formula $\Apath\Next\varphi$ such that
$\truthvalue{\Apath\Next\varphi}{M,q} = x$.
Moreover, let $M'$ be the same as $M$ except for the weights of all the outgoing transitions from $q$ being strictly lower than $x$.
Then we have $\truthvalue{\Apath\Next\varphi}{M',q} < \truthvalue{\Apath\Next\varphi}{M,q}$.
Analogous characterizations can be shown for all other temporal operators.

On the other hand, the transition weights can be also interpreted as a qualitative distribution, similarly to quantitative transitions
in Markov chains and Markov decision processes, used in the semantics of probabilistic temporal logics~\cite{Hansson94PCTL,Huth97quantitative,Kwiatkowska02prism} and their strategic variants~\cite{Huang12probabilisticATL,Chen13prismgames,Alfaro05discounted,Jamroga08mtl-aamas,Jamroga08mtl-prima}.
A natural assumption in that case is that the distribution is complete. In the probabilistic case, it amounts to the weights on the outgoing edges from $q$ always summing up to $1$.
The requirement is essential in models of multi-agent systems, where establishing what \emph{cannot} happen is often as important as reasoning about what can.

\medskip
In the qualitative case, at a minimum, $\truthvalue{\varphi}{M,q'} = x$ on all the successors $q'$ of $q$ should imply
$\truthvalue{\Apath\Next\varphi}{M,q} = x$ (and analogously for other temporal operators and strategic operators).
In particular, if the value of $\varphi$ is bound to be $\top$ at the next moment -- no matter how the systems evolves
 -- then $\Apath\Next\varphi$, $\coop{A}\Next\varphi$, etc., should also evaluate to $\top$.
It is easy to see that the semantics in~\cite{kp02a,kp06} do not satisfy this requirement.

\medskip
In the remainder of the section, we outline how the probabilistic approach can be adapted to arbitrary lattices of transition weights.
Our proposal is based on the concept of \emph{designated paths}, i.e., paths that are considered relevant in a given context.
We also show that the idea of may/must abstraction can be seen as a special, 3-valued case of this kind of reasoning.

\subsection{Weighted transitions in concurrent game structures}\vspace{1mm}

\begin{definition}[Weighted multi-valued CGS]
Assume two lattices: an interpreted lattice $\L^+ = (\LL, \leq, \sigma)$ of truth values,
and a lattice $\L_t= (\LL_t, \leq_t)$ for weights that will be assigned to transitions.
A \emph{weighted multi-valued concurrent game structure (w\mvmodel)} over $\L^+$ and $\L_t$
is a tuple $M = \tuple{\Agt, \States, Act, d, \trans, w, \AP, \PVal,\L^+, \L_t}$,
where $\Agt$, $\States$, $Act$,  $\trans$, $\AP$ are as in case of an mv-CGS, and
$w: \trans \rightarrow \LL_t$ is a weight function which maps each individual transition
in $\trans$ (i.e., each tuple $(q,\alpha_1,\dots,\alpha_k,\trans(q,\alpha_1,\dots,\alpha_k))$ to a value in $\LL_t$.
\end{definition}

\eject
\noindent
The interpretation of \ATLSi formulas in a w\mvmodel $M$ as above is parameterized
by the set ${\cal D}$ of designated values in  $\L_t$ --- the logical values whose assignment to a formula makes
it deemed to be satisfied.

A path in an w\mvmodel is defined analogously as in an \mvmodel.
A path $\lambda=q_0q_1q_2\dots$ is said to be designated if for every $i$
there are actions $\alpha_1,\dots,\alpha_k$ such that $t(q_i,\alpha_1,\dots,\alpha_k)= q_{i+1}$ and
$w(q_i,\alpha_1,\dots,\alpha_k,q_{i+1})\in {\cal D}$.

\medskip
Given ${\cal D}$, we reduce a w\mvmodel $M$ to an \mvmodel
\[
M_{\cal D} = \tuple{\Agt, \States, Act, d, \trans_{\cal D}, \AP, \PVal, \L^+}
\]
 where
 \[
 \trans_{\cal D}(q,\alpha_1,\dots,\alpha_k) = \left\{
  \begin{array}{ll}
   \trans(q,\alpha_1,\dots,\alpha_k)   & \mbox{if } w(q,\alpha_1,\dots,\alpha_k,\trans(q,\alpha_1,\dots,\alpha_k)) \in {\cal D}\\
  \mbox{undefined} & \mbox{otherwise}
   \end{array}
   \right.
 \]
 Then any state of $M$ is a state of $M_{\cal D}$, and any designated path in $M$ is a path in $M_{\cal D}$.
 As the interpretation of \ATLSi in $M$ we take the interpretation of \ATLSi in $M_{\cal D}$:

\medskip
For any state or designated path $\xi$ in $M$ and any formula $\vfi$ in \ATLSi, we take:
\begin{equation}
\truthvalue{\varphi}{M,\xi, {\cal D}} = \truthvalue{\varphi}{M_{\cal D},\xi}
\label{eq:designated}
\end{equation}

\subsection{Embedding may/must abstractions}

A natural example of many-valued transitions is provided by may-must transitions:
the ``may'' transitions are only possible,
but need not happen, while the ``must'' transitions will necessarily take place.
This kind of models were used in~\cite{Godefroid02abstraction}.
Also, the may/must abstractions presented in~\cite{Ball06mv-AMC,LomuscioM15,Belardinelli17abstraction}
produce models with the same or similar structure and interpretation.
Adapting
the notation, those models take the form $M^{GJ} = \tuple{\States, \AP, \delta_{must}, \delta_{may}, V, \L}$,
where $\States, \AP, V$ are defined as before,
$\L = \{\ttt, \ff, \uu\}$, and  $\delta_{must}$, $\delta_{may}\subseteq \States \times \States$ are transition relations
such that $\delta_{must} \subseteq \delta_{may}$.
The language contains negation and conjunction interpreted as in Kleene three-valued calculus over $\{\ttt, \ff,\uu\}$,
and the $AX$ operator interpreted as:
\[
\begin{array}{ll}
\truthvalue{AX\varphi}{M^{GJ},q} = \left\{
                            \begin{array}{ll}
                            \ttt & \mbox{if $\forall s' (\delta_{may}(s, s') \Rightarrow  \truthvalue{\varphi}{M^{GJ},s'}=\ttt$})
                            \\
                            \ff & \mbox{if $\exists s' (\delta_{must}(s, s') \wedge  \truthvalue{\varphi}{M^{GJ},s'}=\ff$})
                            \\
                            \uu & \mbox{otherwise}
                            \end{array}
                        \right.

\end{array}
\]
If, following \cite{Godefroid02abstraction}, we disregard explicit inclusion of agents and actions in our approach,
then if the transition relation $\delta_{may}$ is  a function, such a model can be represented as
an  w\mvmodel $M^{JKP} = \tuple{\States, \delta_{may}, \AP, \PVal,\L^+, \L_t}$ with three-valued transitions,
where $\Ltrans = \{\top, U, \bot\}$, and the weight function $w: \States \times \States \rightarrow \Ltrans$ is defined by:
$$
w(s,s') = \left\{
                            \begin{array}{ll}
                            \top  & \mbox{if $(s,s')\in\delta_{must}$}\\
                            U  & \mbox{if $(s,s')\in\delta_{may}\setminus\delta_{must}$}\\
                            \bot & \mbox{otherwise}
                            \end{array}
                            \right.
                            $$
Denote ${\cal D}_{\top} = \{\top\}, {\cal D}_{U} = \{U, \top\}$. We can show that Godefroid's-Jagadessan's
semantics based on $M^{GJ}$ can be expressed using our model $M^{JKP}$  as follows:

\begin{lemma}
 If $\varphi$ does not contain the $AX$ operator, then:
\begin{enumerate}
\item $\truthvalue{\varphi}{M^{GJ}, q} = \truthvalue{\varphi}{M^{JKP}, q, {\cal D}_{\top}}$
\item
$\truthvalue{AX\varphi}{M^{GJ},q} = \left\{
                            \begin{array}{ll}
                            \ttt  & \mbox{if $\truthvalue{AX\varphi}{M^{JKP}, q, {\cal D}_U}= \ttt$}\\
                            \ff & \mbox{if $\truthvalue{AX\varphi}{M^{JKP}, q, {\cal D}_{\top}}=\ff$}\\
                          \uu & \mbox{otherwise}
                            \end{array}
                            \right.
                            $
\end{enumerate}
\end{lemma}

\begin{proof}
Since both $M^{JKP}$ and $M^{GJ}$ are based on Kleene 3-valued calculus of propositional formulas,
Condition 1 obviously holds.
Further, as the translation of $M^{JKP}$ to $M_{D_{U}}$ preserves all transitions in $\delta_{may}$,
we have $\truthvalue{AX\varphi}{M^{JKP}, q, {\cal D}_U}= \ttt$ iff $\truthvalue{\varphi}{M^{JKP}, q', {\cal D}_U}= \ttt$
for every $(q,q')\in \delta_{may}$. In view of Condition 1, the latter implies
$\truthvalue{\varphi}{M^{GJ}, q'}= \ttt$ for every $(q,q')\in \delta_{may}$.
Consequently, the first clause in Condition 2 holds. For the second clause, note that as $M_{D_{\top}}$
only contains transitions in $\delta_{must}$, then $\truthvalue{AX\varphi}{M^{JKP}, q, {\cal D}_{\top}}=\ff$
iff there is a transition $(q,q')\in \delta_{must}$ such that
$\truthvalue{\varphi}{M^{JKP}, q'}=\ff$.
Then by Condition 1 $\truthvalue{\varphi}{M^{GJ}, q'}= \ff$, whence $\truthvalue{AX\varphi}{M^{GJ},q} = \ff$ --- and so Condition 2 holds.
\end{proof}

\subsection{Model checking multi-valued CGS with weighted transitions}

Fortunately, the introduction of weighted transition, while enriching our models
and making them  better suited to some practical applications, does not introduce any essential complications
into model-checking compared to \mvmodel's with two-valued transitions.
Thus the results obtained in the latter case carry over to \mvmodel,
and we have the following generalization of the Reduction Theorem~\ref{general}:

\begin{theorem}
Let $\L = (\LL, \leq)$ be an arbitrary finite lattice, $\L_f = (\LL_f, \leq_f)$ a sublattice
of $\L$, and let $f: \LL \rightarrow \LL_f$  a mapping which preserves arbitrary bounds in $\L$.
Furthermore, let
$M = $ $\tuple{\Agt, \States, Act, d, \trans, w, \AP, \V, \L^+}$ be an w\mvmodel
over an interpreted lattice $\L^+ = (\LL, \leq, \sigma)$ over ${\cal C}$,
and let
$M_f = \tuple{\Agt, \States, Act, d, \trans, w_f, \AP, \V_f, (\LL_f, \leq_f, \sigma_f)}$ be the \mvmodel obtained from $M$ by ``clustering'' the truth values in $M$ according to $f$, i.e.:
\begin{enumerate}
\item $\sigma_f(c)= f(\sigma(c))$\  for any $c\in \LConst$,
\item $w_f(\tau) = f(w(\tau)))$ for any $\tau\in t$, and
\item $V_f(p,q) = f(V(p,q))$\ for any $q \in \States$ and $p \in \AP$.
\end{enumerate}

Then, for any state (respectively, path) formula $\varphi$ of \ATLSi over $\L$, any state (respectively, path) $\xi$, and any set of designated truth values $\designated$, we have
\begin{equation}
\label{thesis-weighted}
\truthvalue{\varphi}{M,\xi,\designated} \in f^{-1}(x)\qquad \mbox{iff}\qquad \truthvalue{\varphi}{M_f,\xi,\designated} = x
\end{equation}
\end{theorem}
\begin{proof}
Straightforward from Equation~(\ref{eq:designated}) and Theorem~\ref{general}.

Note that the conditions of the above theorem (preservation of the bounds plus Conditions 1 and 3) correspond to those of Theorem 5.4 , with an analogous Condition 2 for weights added).
\end{proof}

This theorem can be used, in a way analogous to that employed for \ATLSi with two-valued transitions,
to reduce mv-model checking for \ATLSi\ with mv-transitions to two-valued model checking.
This is because the semantics of \ATLSi\ with many-valued transitions contains an embedded
reduction of models with mv-transitions to models with two-valued transitions ---
and for those models we can again use the reduction  based on our threshold functions.
Consequently, the local and global model checking algorithm given in Figure~\ref{fig:recursive-mcheck}
and Figure~\ref{fig:global-mcheck} also carry-over to the case of many-valued transitions.

Like previously, the positive results quoted above apply to formulas which do not involve
the implication operator. For the formulas involving that operator, the negative result obtained
in case of two-valued transitions of course still holds --- because a CGS with two-valued transitions
is just a special case of a CGS with many-valued transitions.

\section{Multi-valued verification of agents with imperfect information}\label{sec:imperfinfo}

\ATL and \ATLs were originally proposed for reasoning about agents in perfect information scenarios.
It can be argued that realistic multi-agent systems always include some degree of limited observability~\cite{Schobbens04ATL,Jamroga03FAMAS,Agotnes04atel,Jamroga04ATEL,Agotnes06action,Jamroga07constructive-jancl,Schnoor10strategic}.
However, model checking of \ATL and \ATLs with imperfect information is hard -- more precisely, \Deltwo- to \Pspace-complete for agents playing memoryless strategies~\cite{Schobbens04ATL,Jamroga08mcheckcloser,Bulling10verification} and undecidable for agents with perfect recall~\cite{Dima11undecidable}.
Furthermore, the imperfect information semantics of strategic ability does not admit standard fixpoint equivalences~\cite{Bulling14comparing-jaamas}, which makes incremental synthesis of strategies cumbersome.
Practical attempts at the problem have emerged only
recently~\cite{Pilecki14synthesis,Busard14improving,Huang14symbolic-epist,CernakLMM14,Jamroga18fixpApprox-aij,Kurpiewski19domination}, and the experimental results show that verification is feasible only for very small models.

Such hard problems can be potentially tackled by means of approximation techniques~\cite{Belardinelli17abstraction,Jamroga17fixpApprox}. In particular, abstraction techniques~\cite{Cousot77abstraction,Belardinelli17abstraction,Kouvaros17predicateAbstraction} can be used to cluster multiple states and/or transitions in the system into \emph{abstract} states and transitions, thus reducing the model size.
However, in order to be effective, the abstraction must be very coarse, which potentially results in loss of information about the truth of (some) atomic propositions  and the existence of (some) transitions.
This leads to a substantial reduction of the verification cost, possibly at the expense of introducing non-classical truth values of propositions in some abstract states, as well as transitions of various strength.
In consequence, multi-valued model checking can be extremely useful when reasoning about strategies under uncertainty.

Clearly, all the previously cited reasons for using multi-valued verification
(description of the world based on a non-classical notion of truth, lifting the logical reasoning to a richer domain of answers, inconclusive or inconsistent information about the system, conflicting evidence coming from different sources, inconclusive verification procedure, etc.)
are also relevant for agents with uncertainty.
In this section, we show that the framework of \ATLSi can be easily extended to the case of imperfect information.

\subsection{Logic \ATLSi with imperfect information}

Let us extend \mvmodel with epistemic indistinguishability relations $\sim_1,\dots,\sim_k \subseteq \States\times\States$, one per agent in $\Agt$.
The idea is that, whenever $q\sim_a q'$ and the system is in state $q$, agent $a$ might think that the system is actually in $q'$.
Each $\sim_a$ is assumed to be an equivalence relation.
We also assume that the resulting model is \emph{uniform} with respect to the indistinguishability relations,
i.e., $q\sim_a q'$ implies $d_a(q) = d_a(q')$.
In other words, the choices available to an agent are identical in the states indistinguishable for that agent.

In a similar way, strategies under imperfect information must specify identical choices in indistinguishable situations.
That is,
memoryless strategies with imperfect information (\ir strategies, for short) are functions $s_a : \States\to\Actions$ such that $q\sim_a q'$ implies $s_a(q)=s_a(q')$.
Moreover, perfect recall strategies with imperfect information (shortly: \iR strategies) are functions $s_a : \States^+\to\Actions$ st.~$q_0\sim_a q_0', \dots, q_n\sim_a q_n'$ implies $s_a(q_0\dots q_n)=s_a(q_0'\dots q_n')$.
Again, collective strategies for $A\subseteq\Agt$ are tuples of individual strategies for $a\in A$.
We denote them by $\Sigma_A^\ir$ and $\Sigma_A^\iR$, respectively.

\medskip
The semantics of \ATLo[\mathfrak{S}], parameterized by the type of stra\-tegies $\mathfrak{S}=\IR,\Ir,\ir,\iR$,
can be defined by replacing the clause for the strategic operators from Section~\ref{sec:multivalued} as follows:
\begin{center}
$\truthvalue{\coop{A}\gamma}{M,q}^\mathfrak{S}\ =\
  \bigjoin_{s_A\in\Sigma_A^\mathfrak{S}}\bigmeet_{\lambda\in out(q,s_A)}\set{\truthvalue{\gamma}{M,\lambda}^\mathfrak{S}}$;

$\truthvalue{\noavoid{A}\gamma}{M,q}^\mathfrak{S}\ =\
  \bigmeet_{s_A\in\Sigma_A^\mathfrak{S}}\bigjoin_{\lambda\in out(q,s_A)}\set{\truthvalue{\gamma}{M,\lambda}^\mathfrak{S}}$.\vspace*{-2mm}
\end{center}

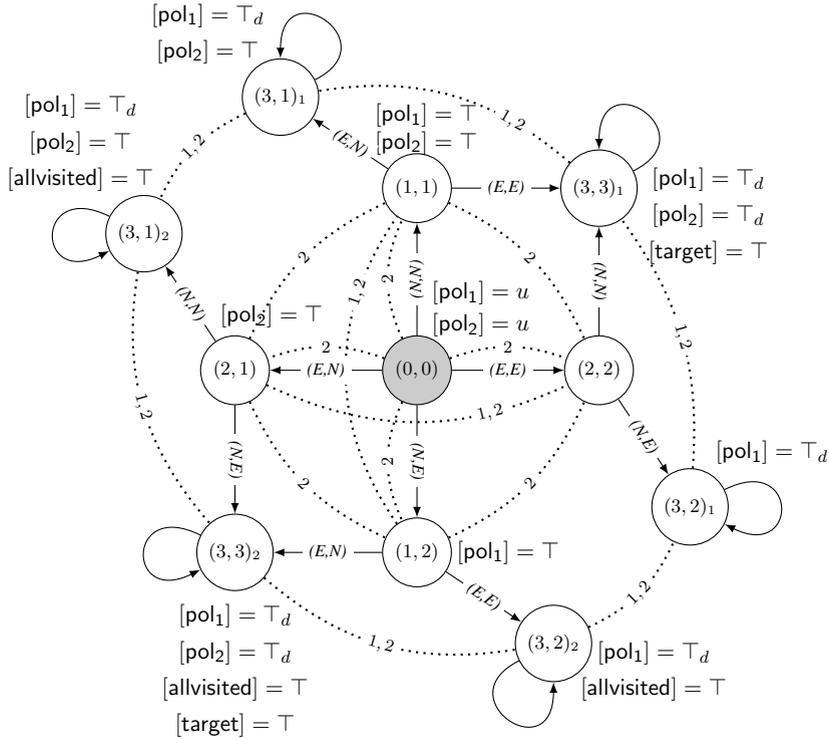
\begin{figure}[!h]
\vspace*{-3mm}
\centering
\hspace{-0.5cm}
\begin{tikzpicture}[>=latex,scale=1.2, every node/.style={scale=0.8}]
  \tikzstyle{state}=[circle,draw,trans, minimum size=8mm]
  \tikzstyle{initstate}=[circle,draw,trans, minimum size=8mm, fill=lightgrey]
  \tikzstyle{trans}=[font=\footnotesize]
  \tikzstyle{epistemic}=[dotted,thick,font=\scriptsize]

  \path (0,0) node[initstate] (q00) {$(0,0)$}
                +(0.7,0.5) node {$[\prop{pol_2}]=\undec$}
                +(0.7,0.85) node {$[\prop{pol_1}]=\undec$}
    (0,2) node[state] (q11) {$(1,1)$}
                +(0.1,0.8) node {$[\prop{pol_1}]=\top$}
                +(0.1,0.5) node {$[\prop{pol_2}]=\top$}
    (2,0) node[state] (q22) {$(2,2)$}		  							  						
    (2,2) node[state] (q33) {$(3,3)_1$}
                +(1.2,0.1) node {$[\prop{pol_1}]=\top_d$}
                +(1.2,-0.3) node {$[\prop{pol_2}]=\top_d$}	  								
                +(1.2,-0.7) node {$[\prop{target}]=\top$}					
    (-2,-2) node[state] (q33v) {$(3,3)_2$}
                +(0,-0.7) node {$[\prop{pol_1}]=\top_d$}
                +(0,-1.1) node {$[\prop{pol_2}]=\top_d$} 			  												
                +(0,-1.5) node {$[\prop{allvisited}]=\top$}					
                +(0,-1.9) node {$[\prop{target}]=\top$}					
    (0,-2) node[state] (q12) {$(1,2)$}
                +(1,0) node {$[\prop{pol_1}]=\top $}
    (-2,0) node[state] (q21) {$(2,1)$}
                +(0.4,0.6) node {$[\prop{pol_2}]=\top$}					  	
    (-3,1.5) node[state] (q31v) {$(3,1)_2$}
                +(-0.7,1.4) node {$[\prop{pol_1}]=\top_d$}
                +(-0.7,1) node {$[\prop{pol_2}]=\top$}
                +(-0.7,0.6) node {$[\prop{allvisited}]=\top$}					
    (-1.5,3) node[state] (q31) {$(3,1)_1$}
                +(-0.8,0.9) node {$[\prop{pol_1}]=\top_d$}
                +(-0.8,0.5) node {$[\prop{pol_2}]=\top$}
    (3,-1.5) node[state] (q32) {$(3,2)_1$}
                +(0.9,0.6) node {$[\prop{pol_1}]=\top_d$}
    (1.5,-3) node[state] (q32v) {$(3,2)_2$}
                +(1.1,-0.1) node {$[\prop{pol_1}]=\top_d$}
                +(1.1,-0.5) node {$[\prop{allvisited}]=\top$}					
      ;

 \path[epistemic] (q00)
       edge[bend left=25]
         node[midway,sloped,rotate=180]{\onlabel{$2$}} (q11)
       edge[bend right=25]
         node[midway,sloped]{\onlabel{$2$}} (q12)
       edge[bend left=20]
         node[midway,sloped]{\onlabel{$2$}} (q22)
       edge[bend right=20]
         node[midway,sloped]{\onlabel{$2$}} (q21);
 \path[epistemic] (q11)
       edge[bend left=20]
         node[midway,sloped]{\onlabel{$2$}} (q22);
 \path[epistemic] (q12)
       edge[bend right=20]
         node[midway,sloped]{\onlabel{$2$}} (q22)
       edge[bend left=35]
         node[near end,sloped]{\onlabel{$1,2$}} (q11);
 \path[epistemic] (q21)
       edge[bend left=20]
         node[midway,sloped]{\onlabel{$2$}} (q11)
       edge[bend right=25]
         node[near end,sloped]{\onlabel{$1,2$}} (q22)
       edge[bend right=20]
         node[midway,sloped]{\onlabel{$2$}} (q12);
 \path[epistemic]
 (q33v) edge[bend right=25]
         node[midway,sloped]{\onlabel{$1,2$}} (q32v)
 (q32v) edge[bend right=20]
         node[midway,sloped]{\onlabel{$1,2$}} (q32)
 (q32) edge[bend right=20]
         node[midway,sloped]{\onlabel{$1,2$}} (q33)
 (q33) edge[bend right=25]
         node[near start,sloped]{\onlabel{$1,2$}} (q31)
 (q31) edge[bend right=20]
         node[midway,sloped]{\onlabel{$1,2$}} (q31v)
 (q31v) edge[bend right=25]
         node[midway,sloped]{\onlabel{$1,2$}} (q33v)
 ;

 \path[->,font=\scriptsize] (q00)
       edge
         node[midway,sloped]{\onlabel{(N,N)}} (q11)
       edge
         node[midway,sloped]{\onlabel{(E,E)}} (q22)
       edge
         node[midway,sloped]{\onlabel{(N,E)}} (q12)
       edge
         node[midway,sloped]{\onlabel{(E,N)}} (q21);
 \path[->,font=\scriptsize] (q11)
       edge
         node[midway,sloped]{\onlabel{(E,E)}} (q33)
       edge
         node[midway,sloped]{\onlabel{(E,N)}} (q31);
 \path[->,font=\scriptsize] (q22)
       edge
         node[midway,sloped]{\onlabel{(N,N)}} (q33)
       edge
         node[midway,sloped]{\onlabel{(N,E)}} (q32);
 \path[->,font=\scriptsize] (q12)
       edge
         node[midway,sloped]{\onlabel{(E,N)}} (q33v)
       edge
         node[midway,sloped]{\onlabel{(E,E)}} (q32v);
 \path[->,font=\scriptsize] (q21)
       edge
         node[midway,sloped]{\onlabel{(N,N)}} (q31v)
       edge
         node[midway,sloped]{\onlabel{(N,E)}} (q33v);

 \draw[-latex,black](q33) ..controls +(1.2,0.6) and +(0,1.4).. (q33);
 \draw[-latex,black](q31) ..controls +(1.2,0.6) and +(0,1.4).. (q31);
 \draw[-latex,black](q32v) ..controls +(-1.2,-0.6) and +(-0,-1.4).. (q32v);
 \draw[-latex,black](q31v) ..controls +(-1.2,0.6) and +(-1.2,-0.6).. (q31v);
 \draw[-latex,black](q32) ..controls +(1.2,0.6) and +(1.2,-0.6).. (q32);
 \draw[-latex,black](q33v) ..controls +(-1.2,0.6) and +(-1.2,-0.6).. (q33v);

\end{tikzpicture}\vspace*{-4mm}
\caption{Multi-valued model $\Mmultiimperf$ for drones with imperfect information. Epistemic indistinguishability is depicted by dotted lines.}
\label{fig:drones-mv-imperfect}\vspace*{-2mm}
\end{figure}

\begin{example}[Drones with partial information]\label{ex:drones-mv-imperfect}
Consider again the drone model introduced in Example~\ref{ex:drones-mv} and Figure~\ref{fig:drones-mv}.
We assume now that drone $1$ sees its own position but not that of drone $2$, whereas drone $2$ only sees if the other drone is in the same location but does not recognize the location itself.
Moreover, each drone can identify the initial state (i.e., $(0,0)$), as well as recognize that it has run out of battery (states $(3,3)_1$ and $(3,3)_2$).
Finally, drone $2$ -- not knowing its exact position -- may try to fly in a direction which is not available for a given location (e.g., fly North in location $2$). In that case, the attempt fails, and the drone stays in its current location.
The updated \mvmodel $\Mmultiimperf$ is presented in Figure~\ref{fig:drones-mv-imperfect}.

\medskip
For the formulas from Example~\ref{ex:drones-mvatl}, we now have :

\begin{itemize}
\item $\truthvalue{\coop{1}\Sometm\pol{1}}{\Mmulti,(0,0)}^\ir = \truthvalue{\coop{1}\Sometm\pol{1}}{\Mmulti,(0,0)}^\iR = \top$, as the strategy to fly North in state $(0,0)$, and then East in $(1,1)$ or $(1,2)$ is uniform for drone $1$;

\smallskip
\item $\truthvalue{\coop{2}\Sometm\pol{2}}{\Mmulti,(0,0)}^\ir = \truthvalue{\coop{2}\Sometm\pol{2}}{\Mmulti,(0,0)}^\iR = \top$ (the analogous strategy for drone $2$ is \emph{not} uniform, but the agent can achieve the goal by playing $N$ in all the states);

\smallskip
\item $\truthvalue{\coop{1,2}\Sometm(\prop{target} \land \prop{allvisited} \land (\pol{1}\lor\pol{2}))}{\Mmulti,(0,0)}^\ir = \bot$ because neither of the uniform memoryless strategies leads to a state where $\prop{target} \land \prop{allvisited}$ holds;

\smallskip
\item $\truthvalue{\coop{1,2}\Sometm(\prop{target} \land \prop{allvisited} \land (\pol{1}\lor\pol{2}))}{\Mmulti,(0,0)}^\iR = \top_d$ (example strategy: drone $1$ flies North in the first step, and East in the second, while drone $2$ moves East and then North).
\finis
\end{itemize}
\end{example}

\para{Objective vs.~subjective semantics of ability.}
We note that the above semantic rule corresponds to the notion of \emph{objective ability}.
That is, given a strategy, we only look at its outcome paths starting from the current global state of the system $q$.
The alternative, \emph{subjective ability}, requires the strategy to succeed on all the paths starting from states indistinguishable from $q$.
Let $\sim_a\!\!(q) = \{q' \mid q \sim_a q'\}$.
This can be formalized by the following adaptation of the semantic rule:

\begin{center}
$\truthvalue{\coop{A}\gamma}{M,q}^\mathfrak{S}\ =\
  \bigjoin_{s_A\in\Sigma_A^\mathfrak{S}}\bigmeet_{a\in A}\bigmeet_{q' \in \sim_a(q)}
	\bigmeet_{\lambda\in out(q',s_A)}\set{\truthvalue{\gamma}{M,\lambda}^\mathfrak{S}}$;

$\truthvalue{\noavoid{A}\gamma}{M,q}^\mathfrak{S}\ =\
  \bigmeet_{s_A\in\Sigma_A^\mathfrak{S}}\bigjoin_{a\in A}\bigjoin_{q' \in \sim_a(q)}
	\bigjoin_{\lambda\in out(q',s_A)}\set{\truthvalue{\gamma}{M,\lambda}^\mathfrak{S}}$.
\end{center}

A more detailed discussion on the epistemic aspects of strategic ability can
be found in~\cite{Agotnes15handbook,Jamroga15specificationMAS}.
We leave the proper treatment of diverse epistemic variants of \ATLSi for the future.

\subsection{Model checking techniques and formal results}

We emphasize again that the correctness of the techniques proposed in Section~\ref{sec:mcheck-mv}
\emph{does not depend on the actual definition of the strategy sets $\Sigma_A$}.
In consequence, the results carry over to the imperfect information case, and the techniques
can be applied \emph{in exactly the same way} to obtain model checking reductions from \ATLo[\mathfrak{S}]
to the corresponding 2-valued cases.
This demonstrates the power of the translation method that can be directly applied to a vast array of possible semantics for \ATLs.
Again, multi-valued verification of \ATLo[\mathfrak{S}] incurs only linear increase in the complexity compared to the 2-valued case.

\section{Case study: Multi-valued verification of the drone model}\label{sec:drones}

Besides the theoretical results discussed in the preceding sections, we present an experimental evaluation
of our approach to verification of strategic abilities.
To this end, we propose a new scalable benchmark based on the running example employed throughout the paper.
We use the CGS template of the team of drones patrolling for pollution in a city (cf.~Examples~\ref{ex:drones-mv} and~\ref{ex:drones-mv-imperfect},
as well as the graphs in Figures~\ref{fig:drones-mv} and~\ref{fig:drones-mv-imperfect}), but with a more complex map to make the study more realistic.
The details and outcomes of the experiments are presented further on in this section.

\subsection{Model description}

The benchmark is an extension of the drone model used in the previous sections.
To recall, we consider a number of drones flying over a fixed area, with each drone modeled as a separate agent.
The map is represented by a directed graph that defines the locations $Loc$ and the paths used by the agents
to move between those locations.
We employ the map shown in Figure~\ref{fig:experimentsMap}.
For the experiments, we assume that the connections between locations are symmetric (i.e., can be traversed both ways),
and hence an undirected graph is a sufficient representation of the map.

\begin{figure}[h]
\vspace{1mm}
\centering
\includegraphics[scale=0.25]{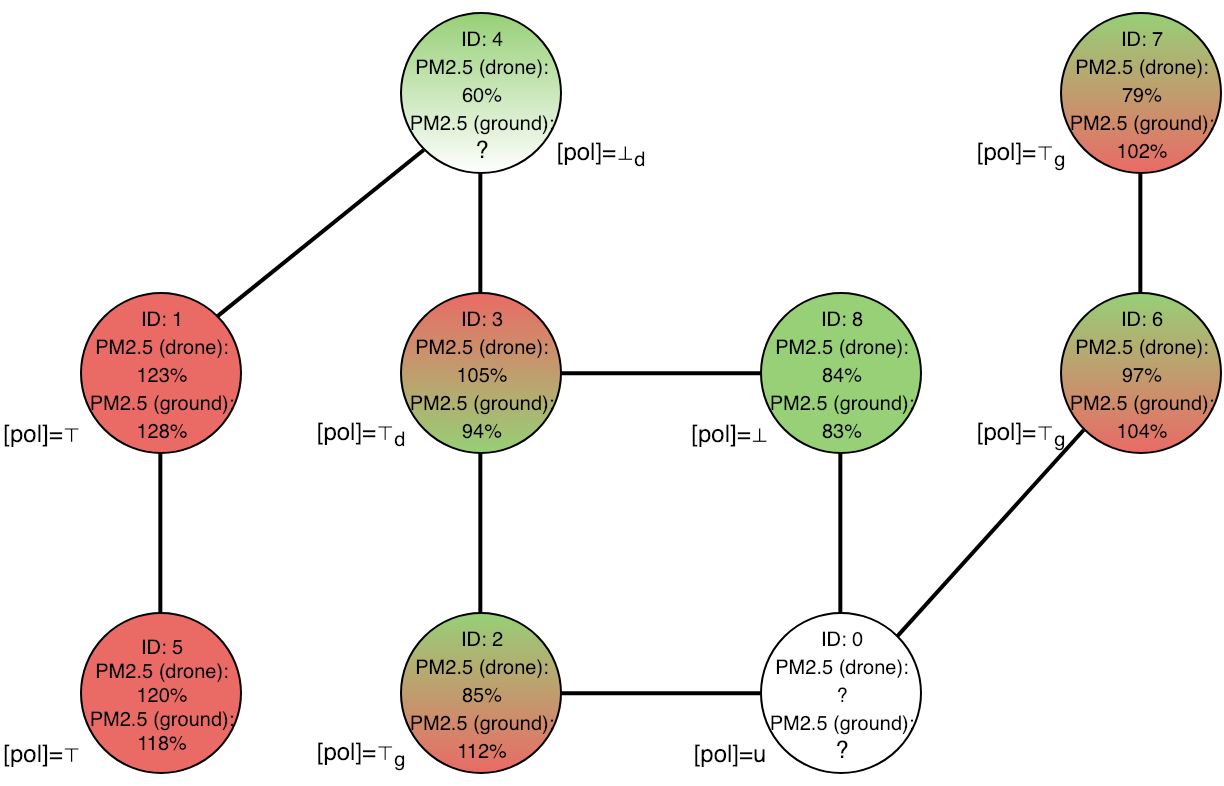}
\caption{The map used in the experiments}
\label{fig:experimentsMap}\vspace*{-2mm}
\end{figure}

The system consists of a number of drone agents and the environment.
The set of all drones is denoted by $D$.
A drone can use its sensors to measure the pollution at its current location.
Moreover, it can communicate with the other drones at the same and adjacent locations using bluetooth,
and obtain their current readings.
The readings from all the ground sensors are broadcasted by the monitoring center, and hence are available
to all drones at all times.
This is modeled by an epistemic indistinguishability relation, with the following information available to the drone:
\begin{itemize}
  \item Its current position (i.e., a location number);
  \item Reading from the drone sensor in its current position;
  \item Readings from the adjacent drones;
  \item Readings from all the ground sensors;
  \item A battery charge level;
  \item A set of already visited places.
\end{itemize}
We assume that the time span of the mission is at most $30$ mins (currently, there are still relatively
few types of drones that can fly longer than a couple of minutes, and they are mostly used in industrial and military contexts).
With this provision, we can assume that the environment is stationary throughout the mission.
That is, while traversing the map the drones will always get the same readings from a given location.

Each drone can perform five possible actions: \emph{go North}, \emph{South}, \emph{East}, \emph{West}, and \emph{Wait}.
Any movement consumes energy.
When the battery level drops to zero, the only action that the drone can perform is \emph{Wait}.
This means it will stay at its current location forever (since our model does not feature battery recharging).
However, such an immobilized drone can still broadcast information to the nearby drones.

As before, we use multi-valued atomic propositions $\prop{pol_d}$, $d\in D$, with values drawn from the lattice $\dlattice$.
The interpretation of $\prop{pol_d}$ is given by the combined readings of drone $d$'s sensor
and of the ground sensor at the current location of $d$.

\medskip
The models are scaled with respect to the following parameters:
\begin{itemize}
	\item Number of drones;
	\item Initial battery level (the same for each drone).
\end{itemize}

\subsection{Formulas}

In the rest of this section, \prop{d} will refer to an arbitrary drone in the set $D$.
The first formula to be verified is
$$\phi_1\quad = \quad \Epath\Sometm \prop{pol_d}\ \ra\ \coop{d}\Sometm \prop{pol_d}$$
It says that if drone $d$ \emph{might} detects pollution to some degree,
then $d$ has a strategy to guarantee that this will indeed be the case.
Note that the formula is an implication, and hence -- due to the results in Section~\ref{sec:impossibility} -- a straightforward reduction to classical model checking is problematic. Because of that, we use the recursive reduction algorithm of Section~\ref{sec:recursive}.
That is, we split $\phi_1$ into its left hand side ($\Epath\Sometm \prop{pol_d}$) and right hand side ($\coop{d}\Sometm \prop{pol_d}$).
We also observe that the left hand side of the implication, expressed in ATL and transformed to the negation normal form, becomes
$\noavoid{\emptyset}\Sometm\prop{pol_d}$.
Thus, in order to determine the value of $\phi_1$, we need to carry out multi-valued model checking of the following two formulas:
\begin{itemize}
\item $\phi_{1L}\ = \noavoid{\emptyset}\Sometm \prop{pol_d}$,\
and
\item $\phi_{1R}\ = \coop{d}\Sometm \prop{pol_d}$,
\end{itemize}
each of them satisfying the preconditions of Theorem~\ref{general}.

\medskip
We observe that the above specification is relatively weak: it requires that if pollution is present somewhere
then the drone is able to find it at some location.
In order to allow for a finer-grained specification, we add to the drone model a family of atomic propositions
$\prop{at_{d,loc}}$ with classical, 2-valued interpretation.
More precisely, $\prop{at_{d,loc}}$ evaluates to $\top$ in the states where drone $d$ is at location $loc \in Loc$,
and to $\bot$ everywhere else.
In addition, we allow for cooperation between the drones.
More exactly, we will be looking at joint strategies of the team of all drones $D$ with the following property:
if any of the drones might detect pollution at location $loc$, then the drones can ensure that one of them will indeed detect it:
$$\phi_2\quad = \quad \bigwedge_{loc \in Loc}\ \big(\Epath\Sometm \bigvee_{d\in D} (\prop{at_{d,loc}} \land \prop{pol_d})
     \quad \ra\quad \coop{D}\Sometm \bigvee_{d\in D} (\prop{at_{d,loc}} \land \prop{pol_d}) \big).$$

Again, the formula is an implication, and thus requires separate treatment of the left and right hand sides of ``$\ra$.''
Here, we only report the verification results for the right-hand subformula, i.e.:
\begin{itemize}
\item $\phi_{2R}^{loc}\ = \coop{D}\Sometm \bigvee_{d\in D} (\prop{at_{d,loc}} \land \prop{pol_d})$
\end{itemize}
for an arbitrary selected value of $loc$.

\medskip
The considered formulas emphasize the importance of the comparison operator $\ra$ for actual specification of properties.
Many (if not most) relevant properties of multi-agent systems are expressed as an implication:
if the assumptions are satisfied, the target property should hold as well.
The multi-valued variant of such requirements demands that $\psi$ is satisfied to at least the same degree as $\varphi$.

\subsection{Semantics and algorithms}

We note that formula $\phi_{1L}$ refers only to the abilities of the empty coalition, and hence does not involve reasoning about imperfect information. In consequence, one can as well evaluate it using the perfect information semantics of \mvATL. Then, the translation in Section~\ref{sec:mcheck-mv-translation} reduces the multi-valued verification of $\phi_{1L}$ to model checking of 2-valued \ATL with perfect information. We implement the latter by means of the standard fixpoint algorithm from~\cite{Alur02ATL}.

\medskip
In contrast, the semantics of formulas $\phi_{1R}$ and $\phi_{2R}$ refers to strategies with imperfect information.
Accordingly, the translation in Section~\ref{sec:mcheck-mv-translation} reduces the problem to model checking of 2-valued \ATL with imperfect information.
Since the exact model checking of abilities under imperfect information is hard, both theoretically~\cite{Schobbens04ATL,Jamroga06atlir-eumas} and in practice~\cite{Lomuscio15mcmas,Busard17phd,Pilecki14synthesis},
we go around the complexity by using the fixpoint-based approximate model checking algorithm proposed recently in~\cite{Jamroga17fixpApprox}.
That is, the 2-valued model checking of $\phi_{1R}$ proceeds by a model-independent translation to its upper and lower variants $\phi_{1R}^U,\phi_{1R}^L$, both of which can be verified by fixpoint algorithms.
If the verification output for $\phi_{1R}^U$ and $\phi_{1R}^L$ matches, it is guaranteed correct for $\phi_{1R}$, too; otherwise, the outcome is inconclusive.
The 2-valued model checking for $\phi_{2R}$ is obtained analogously.

\medskip
As we will see, the output of the lower and the upper approximation always matched in our experiments (cf.~Figures~\ref{fig:experimentsResults-f1R} and~\ref{fig:experimentsResults-f2R}), thus providing fully conclusive outcome.

\subsection{Experimental results}

The results of the experiments are presented in Figures~\ref{fig:experimentsResults-f1L}
(for formula $\phi_{1L}$), \ref{fig:experimentsResults-f1R} (formula $\phi_{1R}$), and~\ref{fig:experimentsResults-f2R} (formula $\phi_{2R}$).
For each of the formulas, we considered several configurations of the drone model.
The main scaling factor was the number of the drones in the system.
The second source of complexity was the initial energy level (the same for every drone).
The initial location on the map was always 0, for all the drones in the system.

\begin{figure}[!h]
\vspace{3mm}
\centering
\scalebox{0.86}{
  \begin{tabular}{|c|Hc!{\vrule width 1pt}c|cHH|c|c|}
    \hline
    
    \#drones & init. place & energy  & \#states  & tgen   & tverif (ir) & output (ir) & tverif & output \\
    \hline
    \hline
    1 & 0 & 1 & 5 & 0.007 & 0.02 & $\top_g$ & 0.01 & $\top_g$ \\ \hline
    1 & 0 & 2 & 14 & 0.005 & 0.06 & $\top_d\join\top_g$ & 0.05 & $\top_d\join\top_g$ \\ \hline
    1 & 0 & 3 & 32 & 0.02 &0.12 & $\top_d\join\top_g$ & 0.10 & $\top_d\join\top_g$ \\ \hline
    1 & 0 & 4 & 61 & 0.04 & 0.23 & $\top$ & 0.26 & $\top$ \\ \hline
    1 & 0 & 5 & 106 & 0.03 & 0.38 & $\top$ & 0.38 & $\top$ \\ \hline
    1 & 0 & 10 & 601 & 0.37 & 2.32 & $\top$ & 2.13 & $\top$ \\ \hline
    1 & 0 & 100 & 16740 & 8.84 & 89.97 & $\top$ & 85.91 & $\top$ \\ \hline
    1 & 0 & 1000 & 178740 & 85.85 & timeout & $-$ & timeout & $-$ \\ \hline
    \hline
    2 & 0 & 1 & 17 & 0.01 & 0.12 & $\top_g$ & 0.08 & $\top_g$ \\ \hline
    2 & 0 & 2 & 98 & 0.05 & 0.57 & $\top_d\join\top_g$ & 0.54 & $\top_d\join\top_g$ \\ \hline
    2 & 0 & 3 & 422 & 0.34 & 2.17 & $\top_d\join\top_g$ & 2.16 & $\top_d\join\top_g$ \\ \hline
    2 & 0 & 4 & 1263 & 1.23 & 6.45 & $\top$ & 6.32 & $\top$ \\ \hline
    2 & 0 & 5 & 3288 & 3.83 & 17.86 & $\top$ & 16.77 & $\top$ \\ \hline
    2 & 0 & 10 & 55757 & 89.23 & 779.57 & $\top$ & 719.81 & $\top$ \\ \hline
    2 & 0 & 100 & $-$ & timeout & $-$ & $-$ & $-$ & $-$ \\ \hline
    \hline
    3 & 0 & 1 & 65 & 0.02 & 0.54 & $\top_g$ & 0.46 & $\top_g$ \\ \hline
    3 & 0 & 2 & 794 & 1.02 & 5.24 & $\top_d\join\top_g$ & 5.34 & $\top_d\join\top_g$ \\ \hline
    3 & 0 & 3 & 6626 & 12.44 & 51.52 & $\top_d\join\top_g$ & 43.09 & $\top_d\join\top_g$  \\ \hline
    3 & 0 & 4 & 31015 & 70.55 & 301.24 & $\top$ & 293.59 & $\top$ \\ \hline
    3 & 0 & 5 & 122140 & 414.57 & 3322.00 & $\top$ & 3003.56 & $\top$ \\ \hline
    3 & 0 & 10 & $-$ & timeout & $-$ & $-$ & $-$ & $-$ \\ \hline
    \end{tabular}
    
}
\caption{Experimental results for $\phi_{1L}$}
\label{fig:experimentsResults-f1L}
\end{figure}

\begin{figure}[H]
\vspace{1mm}
\centering
\scalebox{0.84}{
  \begin{tabular}{|c|Hc!{\vrule width 1.2pt}c|c|c|c|c|c|}
    \hline
    \multirow{2}{*}{\#drones} & \multirow{2}{*}{init. place} & \multirow{2}{*}{energy}  & \multirow{2}{*}{\#states}  & \multirow{2}{*}{tgen}   & \multicolumn{2}{c|}{Lower approx.} & \multicolumn{2}{c|}{Upper approx.} \\
    \cline{6-9}
     & & & & & tverif & output & tverif & output \\
    \hline\hline
    1 & 0 & 1         & 5             & 0.008     & 0.03       & $\top_g$                   & 0.02 & $\top_g$ \\ \hline
    1 & 0 & 2         & 14           & 0.006     & 0.05       & $\top_d\join\top_g$ & 0.06 & $\top_d\join\top_g$ \\ \hline
    1 & 0 & 3         & 32           & 0.01       & 0.14       & $\top_d\join\top_g$ & 0.21 & $\top_d\join\top_g$ \\ \hline
    1 & 0 & 4         & 61           & 0.02       & 0.22       & $\top$ & 0.28 & $\top$                      \\ \hline
    1 & 0 & 5         & 106         & 0.03       & 0.39       & $\top$ & 0.48 & $\top$                      \\ \hline
    1 & 0 & 10       & 601         & 0.27       & 2.12       & $\top$  & 2.38 & $\top$                     \\ \hline
    1 & 0 & 100     & 16740     & 9.70       & 59.27     & $\top$  & 64.40 & $\top$                     \\ \hline
    1 & 0 & 1000   & 178740   & 92.78     & 615.23   & $\top$  & 669.85 &  $\top$                    \\ \hline
    1 & 0 & 10000 & 1798740 & 890.85   & 6125.77 & $\top$  & 6054.80 & $\top$                     \\ \hline
    1 & 0 & 12000 & 2158740 & 1122.84 & timeout  & $-$  & timeout & $-$                           \\ \hline
    \hline
    2 & 0 & 1        & 17            & 0.01       & 0.18       & $\top_g$ & 0.10 & $\top_g$  \\ \hline
    2 & 0 & 2        & 98            & 0.06       & 0.66       & $\top_d\join\top_g$ & 0.64 & $\top_d\join\top_g$ \\ \hline
    2 & 0 & 3        & 422          & 0.33       & 2.63       & $\top_d\join\top_g$ & 2.40 & $\top_d\join\top_g$ \\ \hline
    2 & 0 & 4        & 1263        & 1.21       & 8.43       & $\top$ & 7.31 & $\top$ \\ \hline
    2 & 0 & 5        & 3288        & 3.53       & 19.08     & $\top$ & 18.19 & $\top$ \\ \hline
    2 & 0 & 10      & 55757      & 89.30     & 323.42   & $\top$ & 321.15 & $\top$ \\ \hline
    2 & 0 & 100    & $-$               & timeout  & $-$  & $-$ & $-$ & $-$   \\ \hline
    \hline
    3 & 0 & 1       & 65            & 0.05        & 0.47       & $\top_g$ & 0.52 & $\top_g$ \\ \hline
    3 & 0 & 2       & 794          & 0.95        & 5.71       & $\top_d\join\top_g$ & 5.90 & $\top_d\join\top_g$ \\ \hline
    3 & 0 & 3       & 6626        & 12.31      & 55.60     & $\top_d\join\top_g$ & 50.19 & $\top_d\join\top_g$ \\ \hline
    3 & 0 & 4       & 31015      & 78.63      & 271.55   & $\top$& 239.28 &  $\top$ \\ \hline
    3 & 0 & 5       & 122140    & 382.30    & 1719.81 & $\top$ & 932.31 &  $\top$ \\ \hline
    3 & 0 & 10      & $-$              & timeout    & $-$  & $-$ & $-$ & $-$ \\ \hline
    \end{tabular}
}
\caption{Experimental results for $\phi_{1R}$
}
\label{fig:experimentsResults-f1R} \vspace{6mm}

\centering
\scalebox{0.84}{
  \begin{tabular}{|Hc!{\vrule width 1.2pt}Hc|c|c|c|c|c|}
    \hline
    \multirow{2}{*}{\#drones} & \multirow{2}{*}{energy} & \multirow{2}{*}{loc}  & \multirow{2}{*}{\#states}  & \multirow{2}{*}{tgen}   & \multicolumn{2}{c|}{Lower approx.} & \multicolumn{2}{c|}{Upper approx.} \\
    \cline{6-9}
     & & & & & tverif & output & tverif & output \\
    \hline\hline
    3 & 1 & 7 & 65 & 0.04 & 0.45 & $\bot$ & 0.46 & $\bot$  \\ \hline
    3 & 2 & 7 & 794 & 1.05 & 5.60 & $\top_g$ & 5.43 & $\top_g$  \\ \hline
    3 & 3 & 7 & 6626 & 12.05 & 62.64 & $\top_g$ & 45.45 & $\top_g$  \\ \hline
    3 & 4 & 7 & 31015 & 77.23 & 233.45 & $\top_g$ & 227.44 & $\top_g$  \\ \hline
    3 & 5 & 7 & 122140 & 379.29 & 951.14 & $\top_g$ & 854.41 & $\top_g$  \\ \hline
    3 & 6 & 7 & 349121 & 1276.82 & 2706.15 & $\top_g$ & 2423.55 & $\top_g$  \\ \hline
    3 & 7 & 7 & 880562 & 3446.30 & 6383.17 & $\top_g$ & 6052.17 & $\top_g$  \\ \hline
    3 & 8 & 7 & 1850861 & 9718.62 & timeout & $-$ & timeout & $-$  \\ \hline
    \end{tabular}    
}
\caption{Experimental results: $\phi_{2R}^{loc}$ for $\#drones=3$ and $loc=7$}
\label{fig:experimentsResults-f2R}\vspace*{-3mm}
\end{figure}

\medskip
The experiments were conducted on an Intel Core i7-6700 CPU with dynamic clock speed of 2.60--3.50 GHz, 32 GB RAM, running under 64bit Windows 10.
The times are given in seconds; the timeout was set to 2 hours.
As the performance results show, multi-valued verification of strategic ability scales up similarly to two-valued model checking~\cite{Jamroga17fixpApprox},
which confirms the theoretical results in Section~\ref{sec:recursive}.

The software used to conduct the experiments can be found at the address\\
\texttt{https://github.com/blackbat13/stv}.
The software is implemented in Python 3. As it is an on-going development,
it does not accept any input language. Instead, model generators are used.
Models are generated as transition graphs and stored explicitly in the memory.

\medskip
As the experiments show, the result of the formula depends mostly on the initial energy of the drones. If given enough energy, drones can visit every place on the map, hence detecting any pollution.
On the other hand, even a drone with very small capacity of the battery can detect something.
As can be seen in Figures~\ref{fig:experimentsResults-f1L} and~\ref{fig:experimentsResults-f1R}, for the cases in which initial energy of the drones were less than 3, answer was more informative than simple false, as it would have been if we had used two-valued logic. It shows that multi-valued logics can provide the designer or analyst with much more useful information beyond a simple yes/no answer.

\section{Conclusions}\label{sec:conclusions}

In this paper we study a variant of alternating-time temporal logic, denoted as \ATLSi,
where the truth values are taken from an arbitrary distributive lattice.
We argue that multi-valued model checking of \ATLSi specifications can be useful,
especially for systems whose models cannot be fully analyzed due to their complexity and/or
inaccessibility of the relevant information.
Other examples include systems with information coming from multiple, potentially conflicting sources.
We propose the semantics of \ATLSi first in the simplest case of perfect information strategies and models with crisp,
classical transition functions.
Then, we show how to extend the framework to the case of multi-valued transitions, as well as other notions of strategies
(in particular, variants of strategic reasoning for agents with limited observation capabilities).

In terms of technical results, we prove that our multi-valued semantics of \ATLSi provides a conservative extension of
the classical 2-valued variant.
More importantly, we propose efficient (i.e., polynomial-time) translations from multi-valued model checking to the 2-valued case.
We formally characterize the conditions under which the translation can be carried out by non-recursive one-to-many reduction,
and propose a recursive procedure for the remaining instances of the problem.
The proposed techniques are elegant enough to be directly applicable to other semantic variants of strategic ability,
for example, those referring to imperfect information scenarios.
This allows for non-classical model checking of abilities while benefiting from the ongoing development
of classical model checkers and game solvers.

Finally, we back up our proposal by a series of experiments in a simulated scenario of drones patrolling for pollution in a city.
Besides promising performance results, the experiments demonstrate also the use of the \emph{relevant implication}, based on comparison of truth values, which is among the main novel contributions of this paper.
The operator can be used to provide a multi-valued counterpart of material implication, with an intuitive and appealing interpretation.
This is especially important in multi-agent systems where many relevant properties are indeed based on implication,
which makes them difficult to formalize in the multi-valued case.

In the future, we plan to extend the framework of \ATLSi to richer specification languages,
such as Strategy Logic~\cite{Mogavero10stratLogic,Mogavero14behavioral,Berthon17sl-imperfectinfo}.
We would also like to take a closer look at multi-valued models arising from state and action abstractions,
and to the application of multi-valued model checking to verification of strategic ability under imperfect information.

\medskip\noindent\textbf{Acknowledgements.}
The authors thank Arthur Queffelec for his help in the implementation of the model checking algorithm.
Moreover, Wojciech Jamroga acknowledges the support of the 7th Framework Programme of the European Union under the Marie Curie
IEF project ReVINK (PIEF-GA-2012-626398).
Wojciech Jamroga, Damian Kurpiewski, and Wojciech Penczek acknowledge the support of the National Centre for Research and Development (NCBR),
Poland, under the PolLux projects VoteVerif (POLLUX-IV/1/2016) and STV (POLLUX-VII/1/2019).
Damian Kurpiewski and Wojciech Penczek acknowledge also the support of CNRS/PAN under the project PARTIES.

\end{document}